\documentclass[12pt]{article}
\usepackage{amsmath,amssymb,amsfonts}
\usepackage{amsthm}
\usepackage{float}
\usepackage{bbm}
\usepackage{amsthm}
\usepackage{dsfont}
\usepackage{color}
\usepackage{natbib}
\usepackage{bm}
\usepackage{subcaption}
\usepackage{graphicx}
\usepackage{enumerate}
\usepackage{natbib}
\usepackage{url} 
\usepackage[hyperindex,breaklinks]{hyperref}

\newcommand{\blind}{1}

\newcommand{\Ben}{\begin{enumerate}}
\newcommand{\Een}{\end{enumerate}}
\newcommand{\Bit}{\begin{itemize}}
\newcommand{\Eit}{\end{itemize}}
\newcommand{\Beq}{\begin{equation}}
\newcommand{\Eeq}{\end{equation}}
\newcommand{\Ba}{\begin{align*}}
\newcommand{\Ea}{\end{align*}}

\newcommand{\Mr}{\mathrm}
\newcommand{\Mbb}{\mathbb}

\newcommand{\Z}{\mathbb{Z}}
\newcommand{\N}{\mathbb{N}}

\definecolor{lightgray}{gray}{0.9}

\newcommand{\R}{{\mathbb R}}
\newcommand{\E}{{\mathbb E}}
\renewcommand{\P}{\mathbb{P}}

\newtheorem{Rq}{Remark} 

\theoremstyle{plain}
\newtheorem{theo}{Theorem}
\newtheorem{lemm}{Lemma}

\newtheorem{prop}{Proposition}

\theoremstyle{definition}
\newtheorem{defi}{Definition}
\newtheorem{assu}{Assumption}

\begin{document}

\def\spacingset#1{\renewcommand{\baselinestretch}%
{#1}\small\normalsize} \spacingset{1}


\if1\blind
{
  \title{\bf Forecasting with Markovian max-stable fields in space and time: An application to wind gust speeds}
  \author{Ryan Cotsakis,\thanks{
    The authors gratefully acknowledge the Expertise Center for Climate Extremes (ECCE) at the University of Lausanne for financial support.}\qquad Erwan Koch,\hspace{.2cm}\\
    \small Expertise Center for Climate Extremes (ECCE),\\
    \small Faculty of Business and Economics (HEC) - 
    \small Faculty of Geosciences and Environment (FGSE),\\
    \small University of Lausanne, CH-1015 Lausanne, Switzerland.\\
    \small and \\
    Christian-Yann Robert\\
    \small Laboratory of Actuarial and Financial Science (LSAF), 
    \small Université Lyon 1, Lyon, France.\\
    \small Laboratory in Finance and Insurance (LFA),\\ 
    \small Center for Research in Economics and Statistics (CREST), ENSAE, Paris, France.}
  \maketitle
} \fi

\if0\blind
{
  \bigskip
  \bigskip
  \bigskip
  \begin{center}
    {\LARGE\bf Forecasting with Markovian max-stable fields in space and time: An application to wind gust speeds}
\end{center}
  \medskip
} \fi

\bigskip
\vspace{-1cm}
\begin{abstract}
Hourly maxima of 3-second wind gust speeds are prominent indicators of the severity of wind storms, and accurately forecasting them is thus essential for populations, civil authorities and insurance companies. Space-time max-stable models appear as natural candidates for this, but those explored so far are not suited for forecasting and, more generally, the forecasting literature for max-stable fields is limited. To fill this gap, we consider a specific space-time max-stable model, more precisely a max-autoregressive model with advection, that is well-adapted to model and forecast atmospheric variables. We apply it, as well as our related forecasting strategy, to reanalysis 3-second wind gust data for France in 1999, and show good performance compared to a competitor model. On top of demonstrating the practical relevance of our model, we meticulously study its theoretical properties and show the consistency and asymptotic normality of the space-time pairwise likelihood estimator which is used to calibrate the model.
\end{abstract}

\noindent%
{\it Keywords:} Advection; Brown--Resnick model; Max-autoregressive model; Nowcasting; Space-time max-stable model; Weather forecasting
\vfill

\newpage
\spacingset{1.9} 

\section{Introduction}

Extreme wind events can trigger huge human impacts and are among the most financially devastating natural disasters globally. For instance, the Lothar windstorm in December 1999 resulted in losses exceeding \$8 billion, while in more recent years, individual storms in the south of Europe have caused more than \$4 billion in damage \citep{Gonalves2024}. Producing reliable nowcasts (lead time from $0$ to $6$ hours) as well as short-range (lead time from to $12$ to $72$ hours) and medium range (lead time from three to seven days) forecasts of their evolution is key to issue timely and accurate warnings, and is thus essential for populations, civil authorities, and insurance companies. Existing forecasting strategies include purely observation-based methods (e.g., persistence or analog methods), traditional statistical techniques (using, e.g., autoregressive processes for nowcasting), the use of complex numerical weather prediction (NWP) models \citep{bauer2015quiet}, and recently developed artifical intelligence (AI)-based methods \citep[e.g.,][and references therein]{rasp2024weatherbench}. Although NWP and AI-based approaches often produce the most accurate forecasts at the aforementioned lead times, purely statistical models have the advantages to be interpretable and to allow easy uncertainty quantification. In this paper we leverage spatio-temporal extreme-value
theory (EVT) to propose a parsimonious statistical model which, in addition to
the aforementioned advantages, offers an explicit Markovian representation of the temporal dynamics and thus allows straightforward ensemble forecasting.

We aim at forecasting---in time---hourly maxima of $3$-second wind gust speeds, as they are key indicators of storm severity owing to their damage potential. Such hourly maxima are taken over a large number (1200) of measurements. Despite the strong temporal dependence, the branch of EVT dealing with maxima (block-maxima type of approach) turns out to be appropriate for such data, although often being used for larger blocks (weeks, months or years). Since we are interested in the full spatial field of these hourly data and in its temporal evolution, we need to resort to EVT for pointwise maxima, i.e., the theory of max-stable random fields. Max-stable fields \citep[e.g.,][]{haan1984spectral, de2007extreme, davison2012statistical}, which constitute an extension of multivariate generalized extreme-value random vectors to the functional setting,  indeed naturally arise as limits of properly scaled pointwise maxima. Common models include the Smith \citep{smith1990max}, Schalther \citep{schlather2002models}, Brown--Resnick \citep{brown1977extreme, kabluchko2009spectral} and extremal-$t$ \citep{opitz2013extremal} fields. In order to perform reliable forecasts, we have to suitably model the temporal dynamics, which requires us to be in a space-time setting. \cite{davis2013max}, \cite{huser2014space}, and \cite{buhl2015anisotropic} constructed space-time max-stable models by considering a $d$-dimensional max-stable models and labeling one of the equivalent dimensions as temporal, and the other $d-1$ dimensions as spatial. One limitation of this strategy is that the temporal dynamics exhibit the same structure as the spatial dependence, making the temporal dynamics possibly inadequate, unexplicit, and difficult to interpret. Moreover, forecasting using these models is difficult since the forecast at a future time point involves a conditional (on all observations in a half-space of $\R^{d}$) distribution which is often intractable and difficult to sample from. 

More generally, forecasting with max-stable random fields presents significant theoretical challenges since their conditional distributions are typically intractable, and the associated literature \citep{davis1989basic, davis1993prediction, cooley2007prediction, lebedev2009nonlinear, qian2022class, tang2021uncertain, wang2011conditional} primarily focuses on spatial-only or temporal-only settings. None of these approaches directly addresses the fundamental challenge of temporal forecasting in a way that can be practically applied to our setting with two spatial dimensions and one temporal dimension.

The broad class of models proposed by \cite{embrechts2016space} addresses some of the limitations of the three aforementioned space-time max-stable models. We utilize a specific model from this class that allows easy forecasting due to its Markov property in time and which is well-suited to atmospheric applications due to the presence of an advection parameter that can model propagation of air masses. This enables us to address a significant concrete weather-related problem while filling a gap in the literature dedicated to forecasting with max-stable fields. Although this model was introduced by \cite{embrechts2016space} along with some properties and a brief simulation-based study of the space-time maximum pairwise likelihood estimator, its detailed theoretical properties, the associated forecasting strategy and its practical usefulness for real-life problems remain unexplored.

Our contribution is threefold. First, we provide a detailed study of the model's properties and establish the strong consistency and asymptotic normality of the space-time maximum pairwise likelihood estimator as both spatial and temporal dimensions approach infinity. Second, we develop a novel methodology for forecasting using this model. Finally, we demonstrate our model's practical utility through an application to wind gust speeds over Northwestern France in 1999, showing superior performance compared to the model by \cite{davis2013max}. Our approach exhibits better skill both in capturing the genuine temporal evolution of the field and in producing accurate forecasts, as evidenced by a more realistic representation of the space-time correlation structure and improved forecasting scores. These improvements stem from the explicit temporal dynamics through the Markov property and the advection component, in contrast to the implicit temporal structure in \cite{davis2013max}'s approach.

The remainder of the paper is organized as follows. Section \ref{Sec_Data} describes the data we consider and provides a brief reminder about max-stable random fields. Then, we present the model and our forecasting strategy in Section \ref{Sec_ModelForecasting}. Section \ref{Sec_InferenceSimulations} details the estimation procedure and provides asymptotic properties of the pairwise likelihood estimator. We apply our model to the mentioned dataset in Section \ref{Sec_CaseStudy}. Finally, Section \Ref{Sec_Discussion} provides a summary of our results as well as some perspectives. The supplementary material (Sections~\ref{App_Operational}--\ref{sec:single_site_params}, provided separately) gathers an explanation of how to use our model for operational weather forecasts, proofs, simulation experiments, and some diagnostics. Throughout the paper, $\overset{d}{=}$ and $\overset{d}{\rightarrow}$ denote equality and convergence in distribution, respectively; in the case of random fields, the distribution should be understood as the set of all finite-dimensional multivariate distributions. Moreover, $\overset{\mathrm{a.s.}}{\longrightarrow}$ denotes almost sure convergence. In the following, ``$\bigvee$'' denotes the supremum when applied to a countable set.

\section{Data and preliminaries}
\label{Sec_Data}



\subsection{Data}\label{Sec_Data_Data}


We focus on hourly maxima of wind gust data taken every three seconds. The measurements are taken at $10$ m height (as defined by the World Meteorological Organization) from 19 December 1999 05:00 central European time (CET) to 23 December 1999 13:00 CET over a rectangle domain extending from $-1^{\circ}$ to $3.25^{\circ}$ longitude and $46.5^{\circ}$ to $49.25^{\circ}$ latitude (see Figure \ref{Fig_ConsideredRegion}); 
the spatial resolution is $0.25^{\circ}$ longitude and $0.25^{\circ}$ latitude. We thus have $105$ temporal observations at each of the $216$ grid points. The data were obtained from the publicly available ERA5 (European Centre
for Medium-Range Weather Forecasts Reanalysis
$5^{th}$ Generation) dataset\footnote{\url{https://cds.climate.copernicus.eu/datasets/reanalysis-era5-single-levels?tab=download}}; more precisely we used the ``10 m wind gust since previous post-processing'' variable.

\begin{figure}[!t]
\centering 
\includegraphics[scale = 0.5]{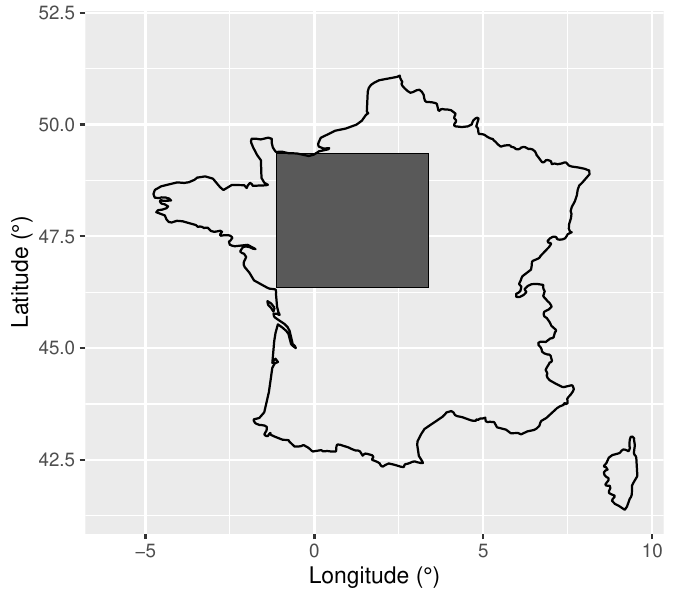}
\caption{Considered region (indicated by the shaded rectangle).}
\label{Fig_ConsideredRegion}
\end{figure}

Figure~\ref{Fig_Snapshots} clearly shows that, from one hour to the next, the main spatial patterns propagate to the East/South-East, which is classic during wet winter periods, where a westerly regime is prevailing. Spatial propagation (advection) takes place for many atmospheric variables (temperature, rainfall, pollutant concentration) and around the world.

\begin{figure}[h!]
\centering 
\includegraphics[scale = 0.6]{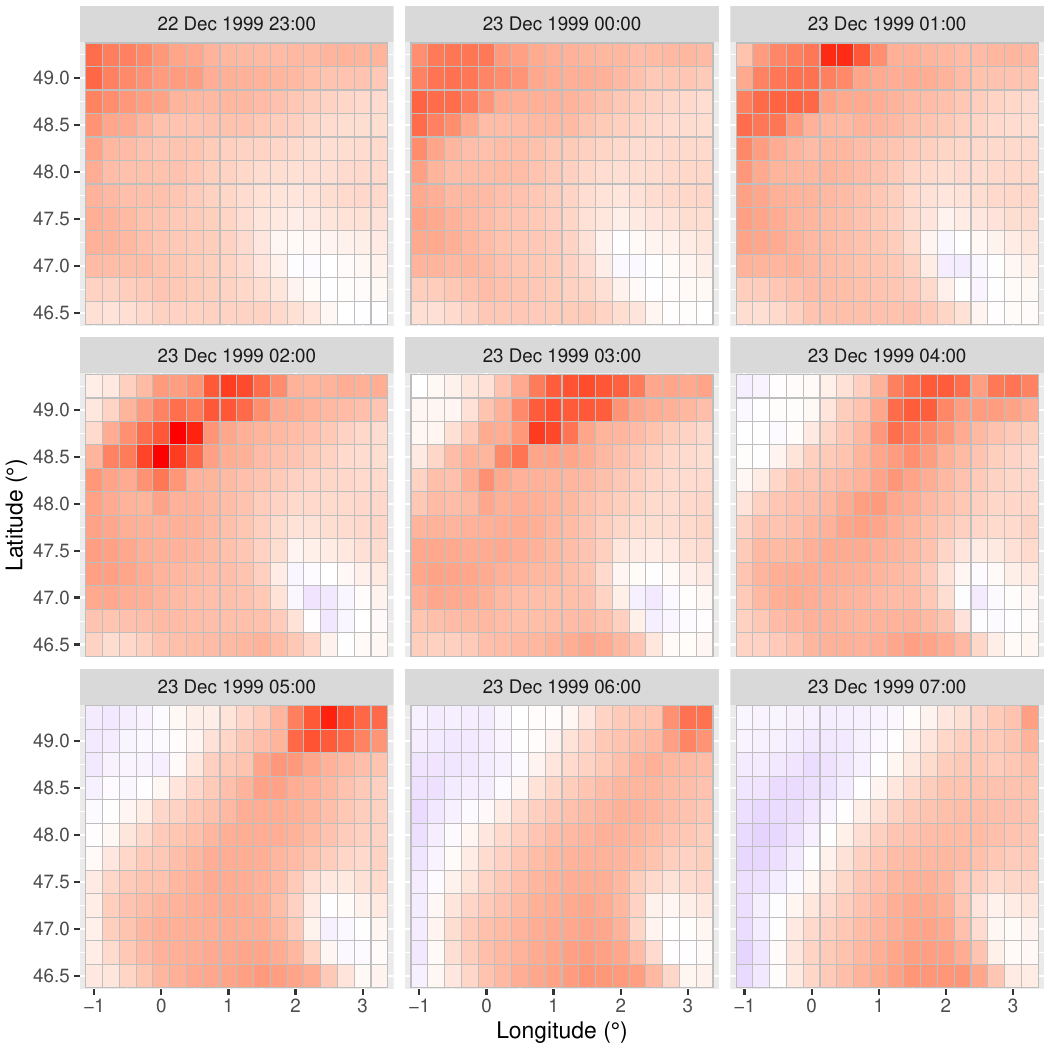}
\caption{Fields of pointwise hourly maxima of $3$-second wind gust on the considered region from 22 December 1999 at 23:00 CET to 23 December 1999 at 07:00 CET. The data have been transformed to follow the standard Gumbel distribution (see Section~\ref{subsect_remindermaxstable}) at each grid point. Blue, white, and red indicate negative, zero, and positive values, respectively. The colour intensity is proportional to the values, with darkest blue corresponding to $-1.22$ m s$^{-1}$ and darkest red to $6.63$ m s$^{-1}$ (on the Gumbel scale).}
\label{Fig_Snapshots}
\end{figure}

\subsection{Reminder about max-stable random fields}
\label{subsect_remindermaxstable}

Let $S_1, \ldots, S_n$ be independent replications of a  random field $\{ S(\bm{x})  \}_{\bm{x} \in \mathbb{R}^d}$, and $(a_n(\bm{x}), \bm{x} \in \mathbb{R}^d)_{n \geq 1}> 0$ and $(b_n(\bm{x}), \bm{x} \in \mathbb{R}^d)_{n \geq 1} \in \mathbb{R}$ be sequences of functions. If there exists a non-degenerate random field $\{ X(\bm{x})  \}_{\bm{x} \in \mathbb{R}^d}$ such that,
$$
\left \{ \frac{ \bigvee_{i=1}^{n} S_i(\bm{x}) -b_n(\bm{x} )}{a_n(\bm{x} )}  \right \}_{\bm{x} \in \mathbb{R}^d } \overset{d}{\to} \left \{ X(\bm{x}) \right \}_{\bm{x} \in \mathbb{R}^d},
$$
then $X$ is necessarily max-stable \citep{haan1984spectral}, which explains the relevance of max-stable fields as models for pointwise maxima of random fields. 
 
Max-stable fields having standard Fréchet margins, i.e., such that $\mathbb{P}(Z(\bm{x}) \leq z) = \exp(-1/z)$, $z>0$, $\bm{x} \in \mathbb{R}^d$, are said to be simple. Sometimes, max-stable fields are also standardized to have Gumbel margins (as, e.g., in Figure~\ref{Fig_Snapshots}), whose distribution function is $\exp(-\exp(-x))$, $x \in \mathbb{R}$.
If $\{ X(\bm{x}) \}_{\bm{x} \in \mathbb{R}^d}$ is max-stable, there exist deterministic functions $\mu(\cdot) \in \Mbb{R}$, $\sigma(\cdot)>0$ and $\xi(\cdot) \in \Mbb{R}$ defined on $\Mbb{R}^d$, called the
location, scale and shape functions, such that
\Beq
\label{Eq_Link_Maxstb_Simple_Maxstab}
X(\bm{x}) = 
\left \{
\begin{array}{ll}
\mu(\bm{x})-\sigma(\bm{x})/\xi(\bm{x}) + \sigma (\bm{x})Z(\bm{x})^{\xi(\bm{x})}/\xi(\bm{x}), & \quad \xi(\bm{x}) \neq 0, \\
\mu(\bm{x}) + \sigma \log Z(\bm{x}), & \quad \xi(\bm{x}) = 0,
\end{array}
\right.
\Eeq
where $\{ Z(\bm{x}) \}_{\bm{x} \in \mathbb{R}^d}$ is simple max-stable. This comes from the fact that, for any $\bm{x} \in \mathbb{R}^d$, $X(\bm{x})$ follows the generalized extreme-value (GEV) distribution with location, scale, and shape parameters  
$\mu(\bm{x})$, $\sigma(\bm{x})$, and $\xi(\bm{x})$. 

Any simple max-stable field can be written as \citep[][]{haan1984spectral}
\begin{equation}\label{eq2}
Z(\bm{x}) = \bigvee_{i=1}^{\infty} U_i Y_i(\bm{x}), \quad \bm{x} \in \mathbb{R}^d,
\end{equation}
where the $(U_i)_{i \geq 1}$ are the points of a Poisson point process on $(0, \infty)$ with intensity function $u^{-2} \mathrm{d}u$ and each $(Y_i)_{i \geq 1}$ is an independent replicate of a non-negative random field $Y$ on $\Mbb{R}^d$, such that $\Mbb{E}[Y(\bm{x})]=1$ for any $\bm{x} \in \mathbb{R}^d$. Additionally, any field defined by \eqref{eq2} is simple max-stable, and this has enabled the construction of parametric models of max-stable fields. 

The best known are the Smith \citep{smith1990max}, Schlather \citep{schlather2002models}, Brown--Resnick \citep{brown1977extreme, kabluchko2009stationary}, and extremal-$t$ \citep{opitz2013extremal} models; the last two have been found to be flexible models that capture environmental extremes well.
Write 
$Y(\bm{x}) = \exp \left\{ \epsilon(\bm{x}) - \mathrm{Var}(\epsilon(\bm{x}))/2
\right\}$, $\bm{x} \in \mathbb{R}^d$, where $\mathrm{Var}$ denotes variance, and $\{ \epsilon(\bm{x}) : \bm{x} \in \mathbb{R}^d \}$ is a centred Gaussian random field with stationary\footnote{Throughout, stationarity refers to strict stationarity, i.e., all finite-dimensional margins are invariant by a shift in space and/or time.} increments and semivariogram $\gamma$ \citep[see, e.g.,][]{matheron1963principles}. Taking this $Y$ in \eqref{eq2} leads to the Brown--Resnick random field associated with the semivariogram $\gamma$. A frequently used isotropic semivariogram is
$ \gamma(\bm{x}) = \left( \| \bm{x} \|/\kappa \right)^{2H}$, $\bm{x} \in \mathbb{R}^d$, where $\kappa > 0$ and $H \in (0, 1]$ are the range and Hurst parameters, respectively, and $\|.\|$ is the Euclidean distance. Note that twice the Hurst index is often referred to as the smoothness parameter.

For any simple max-stable field $Z$ and $\bm{x}_1, \ldots, \bm{x}_D \in \R^d$, we have
\begin{equation}
\label{Eq_Multivariatedf_Maxstable}
    \mathbb{P}(Z(\bm{x}_1) \leq z_1, \ldots, Z(\bm{x}_D) \leq z_D)) = \exp(-V_{Z; \bm{x}_1, \ldots, \bm{x}_D}(z_1, \ldots, z_D)),
    \end{equation}
where $V_{Z; \bm{x}_1, \ldots, \bm{x}_D}$ is the exponent measure of the random vector $(Z(\bm{x}_1), \ldots, Z(\bm{x}_D))^{\prime}$ \citep[e.g.,][]{de2007extreme}, with $^{\prime}$ denoting transposition.
The $D$-dimensional multivariate density of a max-stable vector is often intractable as the exponent measure is difficult to characterize unless $D$ is small, and the exponential leads to a combinatorial explosion of the number of terms in the density. Thus, it is common to estimate max-stable fields using the composite likelihood, most often the pairwise likelihood \citep[e.g.,][]{padoan2010likelihood}.

The bivariate extremal coefficient function for a simple max-stable field $Z$ is defined, for $z>0$, by 
$\mathbb{P}(Z(\bm{x}_1) \leq z, Z(\bm{x}_2) \leq z)) = \exp\left(- \Theta(\bm{x}_1, \bm{x}_2) / z \right)$, $\bm{x}_1, \bm{x}_2 \in \mathbb{R}^d$.
If $Z$ is stationary, then $\Theta$ depends on the lag vector $\bm{h}=\bm{x}_2-\bm{x}_1$ only.

\medskip

Apart from \cite{embrechts2016space}, the main approach \citep[e.g.,][]{davis2013max, huser2014space} used so far to build models for space-time max-stable fields consists of using Representation \eqref{eq2}, noting that $\mathbb{R}^{d} = \mathbb{R}^{d-1} \times \mathbb{R}$, and assigning $\mathbb{R}^{d-1}$ to space and $\mathbb{R}$ to time, i.e., writing $\bm{x}=(\bm{s}, t)^{\prime}$ where $\bm{s} \in \mathbb{R}^{d-1}$ denotes the spatial index and $t \in \mathbb{R}$ denotes the time index. In the space-time setting, very commonly and as is the case in this paper, we consider space to be $2$-dimensional, and \eqref{eq2} therefore becomes 
\begin{equation}\label{Eq_SpaceTimeBR}
Z(\bm{s}, t)=\bigvee_{i=1}^{\infty} U_i Y_i(\bm{s}, t), \quad \bm{s} \in\mathbb{R}^2, t \in \mathbb{R}.
\end{equation}
In this context, the space-time Brown--Resnick field introduced by \citet{davis2013max}, referred to as the DKS model in the following, is given by~\eqref{Eq_SpaceTimeBR}, with $t\geq 0$ and with each $Y_i$ an independent replication of $Y(\bm{s}, t) = \exp \left\{ \epsilon(\bm{s}, t) - \mathrm{Var}(\epsilon(\bm{s}, t))/2
\right\}$, where $\epsilon$ is a space-time Gaussian random field with stationary increments.

Returning to \eqref{Eq_Link_Maxstb_Simple_Maxstab} in the space-time context, we will consider temporal stationarity and thus define the functions $\eta$, $\tau$ and $\xi$ as  functions of space only: $\eta(\bm{s})$, $\tau(\bm{s}) $, and $\xi(\bm{s})$. 

\section{Model and forecasting}
\label{Sec_ModelForecasting}

\subsection{A space-time max-autoregressive model with advection}
\label{Subsec_Model}

We now present the model that will be used in our application to wind gust reanalysis data, once these have been transformed to the standard Fréchet scale. The model is a max-autoregressive space-time max-stable field \citep[belonging to the class introduced by][]{embrechts2016space} with an advection component, and is therefore suitable for forecasting and accommodates the spatial propagation often observed in atmospheric phenomena. Note that max-autoregressive models are the only space-time max-stable models to exhibit the Markovian property in time. The presented model handles the spatio-temporal dependence in the ``standardized (Fréchet) world'' and the marginal distribution at each grid point thus also needs to be modeled; it can be a GEV distribution with parameters specific to that grid point, belonging to a trend surface, or any other distribution depending on the purpose. 

Let $(U_i)_{i \geq 1}$ be the points of a Poisson point process on $(0, \infty)$ with intensity function $u^{-2} \Mr{d}u$ and $(Y_i)_{i \geq 1}$ be independent replications of a non-negative random field $Y$ on $\Mbb{R}^2$ such that $\Mbb{E}[Y(\bm{s})]=1$ for any $\bm{s} \in \mathbb{R}^2$. We consider a parametric spatial max-stable (written as in \eqref{eq2} but in the spatial setting) random field  
\begin{equation}
\label{Eq_InnovationField}
    W(\bm{s}) = \bigvee_{i=1}^{\infty} U_i Y_i(\bm{s}),\qquad \bm{s}\in\R^2,
\end{equation}
where the distribution of the $Y$ field is assumed to depend on a parameter that we denote by $\bm{\theta}$; e.g., $W$ can be a spatial Schlather or Brown--Resnick model. We also introduce a family $(W_t(\bm{s}))_{t \in \mathbb{N}}$ of independent replications of $W$. Our space-time max-stable model $Z$ is then defined as follows:
\Ben
\item \textit{Initialization}: $Z(\bm{s}, 0)=W_0(\bm{s}), \quad \bm{s} \in \Mbb{R}^2$.
\item \textit{Recurrence equation}: for any $t \in \Mbb{N}^+$,
\begin{equation}\label{eqn:max_auto}
    Z(\bm{s}, t) = \max\big\{ a Z(\bm{s}-\bm{\tau}, t - 1),\ (1-a)W_t(\bm{s})\big\}, \quad \bm{s} \in \Mbb{R}^2,
\end{equation}
where $a \in (0, 1)$ and $\bm{\tau} \in \Mbb{R}^2$.
\Een
By $\mathbb{N}^+$ we mean $\mathbb{N} \backslash \{ 0 \}$. This model fundamentally differs in spirit from the space-time max-stable models developed in \cite{davis2013max}, \cite{huser2014space}, and \cite{buhl2015anisotropic}, owing to its explicit dynamics and its causal representation. As already mentioned, it is a max-autoregressive random field; the value of $Z$ at time $t$ and site $\bm{s}$ either corresponds to an attenuated value of the realization of $Z$ at site $\bm{s} - \bm{\tau}$ and time $t-1$ or to a scaled version of the realization of the innovation field $W_t(\bm s)$.
The parameter $a$ governs the strength of influence of the past and is related to the rate at which dependence decays in time. The parameter $\bm{\tau}$ creates a propagation of the spatial patterns with time and thus allows one to capture advection which is an essential feature of atmospheric phenomena (see, e.g., Figure \ref{Fig_Snapshots} in the case of wind gust speed); $\bm{\tau}$ can therefore be seen as a velocity vector. Contrary to $a$ and $\bm \tau$ which control the dynamics, the remaining parameters of $Z$, gathered in $\bm{\theta}$, are inherited from the parametrization of $W$ and only characterize the spatial dependence structure. For any $t \in \mathbb{N}$, the spatial field $\{Z(\bm{s}, t)\}_{\bm{s} \in \mathbb{R}^2}$ is max-stable with the same distribution as that of $W$ \citep[see][Section 3.1]{embrechts2016space}.

It is worthwhile to note that, for any $t\in \N^+$ and $u \in \{1,\ldots,t\}$,
\begin{equation}
\label{Eq_RecurrenceLagu}
    Z(\bm{s}, t) = \max\big\{ a^u Z(\bm{s}-u\bm{\tau}, t - u),\ (1-a^u)\widetilde W_{t-u}^t(\bm{s})\big\}, \quad \bm{s} \in \Mbb{R}^2,
\end{equation}
where for $t_1, t_2 \in \N$, with $t_1 < t_2$,
\begin{equation}\label{eqn:W_tilde}
\widetilde W_{t_1}^{t_2}(\bm s) = \frac{1-a}{1-a^{t_2-t_1}}\bigvee_{k=0}^{t_2-t_1-1}a^k W_{t_2-k}(\bm s - k\bm \tau),\qquad \bm s \in \R^2.
\end{equation}
In addition, for $t_1 < t_2 < t_3$ in $\N$, the random fields $\widetilde W_{t_1}^{t_2}$ and $\widetilde W_{t_2}^{t_3}$ are independent and equal in distribution to $W$. This statement, together with~\eqref{Eq_RecurrenceLagu}, are proved in Section~\ref{Sec:proof_of_closed_form}. Equation~\eqref{Eq_RecurrenceLagu} turns out to be useful for the forecasting (described in Section~\ref{Sec_ModelForecasting}) as it provides the same recurrence pattern as~\eqref{eqn:max_auto} for time steps larger than unity.

By construction, our model is time-Markovian, meaning that the conditional distribution of $Z(\bm{s}, t+u)$ given $\{Z(\tilde{\bm{s}}, \tilde{t}): \tilde{\bm{s}} \in \mathbb{R}^2, \tilde{t} \in \{ 1, \ldots, t \}\}$ is the same as that of $Z(\bm{s}, t+u)$ given $\{Z(\tilde{\bm{s}}, t): \tilde{\bm{s}} \in \mathbb{R}^2 \}$. The Markovian property is even stronger in that, for any $u \in \N^{+}$, the distribution of $Z(\bm{s}, t+u)$ given $\{Z(\tilde{\bm{s}}, t): \tilde{\bm{s}} \in \mathbb{R}^2 \}$ is the same as that of $Z(\bm{s}, t+u)$ given $Z(\bm{s} - u\bm{\tau}, t)$, 
i.e., the only relevant information to forecast $Z(\bm{s}, t+u)$ is $Z(\bm{s} - u\bm{\tau}, t)$.  

As shown in \cite{embrechts2016space}, the space-time field $Z$ defined above is stationary in time. If, in addition, $W$ is stationary in space, which we assume in the following, then $Z$ is stationary in space and time.
According to Proposition 1 in \cite{embrechts2016space}, the bivariate exponent measure of the space-time field $Z$ is written, for $\bm{h} \in \R^{2}$, $u \in \N$,
\begin{align}\label{eqn:bivariate_exponent_measure_general}
    V_{Z; \bm{h}, u}(z_1,z_2) = &\ -\log\P\big(Z(\bm{0},0) \leq z_1, Z(\bm{h},u) \leq z_2\big)\nonumber\\
    = &\ V_{W;\bm{h} - u\bm{\tau}}(z_1, a^{-u}z_2) + \frac{1-a^u}{z_2}, \quad z_1, z_2 >0,
\end{align}
where 
$V_{W;\bm{h}}(z_1,z_2) = -\log\P\big(W(\bm{0}) \leq z_1, W(\bm{h}) \leq z_2\big)$, for $z_1, z_2 >0$,
is the corresponding bivariate exponent measure of $W$.
Hence the expanded expression of $V_{Z;\bm h, u}$ depends on the choice of the innovation field $W$. If $\bm{h}-u\bm{\tau} = 0$, $W$ does not appear in the exponent measures, and one obtains
\begin{equation}
\label{eqn:bivariate_exponent_measure_general_degeneratecase}
    V_{Z;u\bm\tau,u}(z_1,z_2) = \frac{1}{\min \{ z_1,a^{-u}z_2 \} } +\frac{1-a^u}{z_2},\qquad u \in \N.
\end{equation}
The bivariate extremal coefficient, which is defined by
$\Theta_Z(\bm{h},u)=V_{Z;\bm h, u}(1,1)$, for $\bm h\in \R^2, u \in \N$, takes the specific form $\Theta_Z(u\bm\tau,u) = 2 - a^u$ in the case where $\bm{h}=u\bm\tau$.

Provided that the spatial field $W$ is mixing (in the sense of Definition 2.1 in \cite{kabluchko2010ergodic}), our model $Z$ is shown to be space-time mixing in Lemma~\ref{lem:X_mixing}~ (see
Section~\ref{sec:pwle_consistency}), which allows us to establish the asymptotic properties of the pairwise maximum likelihood estimator (see Section~\ref{Sec_InferenceSimulations}) when $W$ takes the form of the spatial Brown--Resnick model.

Owing to the time Markovian property, the information about future time points is entirely described by the current state of the field. It is thus relatively straightforward to express the distribution of the field at a later space-time point conditionally on the value taken at an appropriately chosen spatial point at the current time. Indeed, the conditional distribution of $Z(\bm{s},t+u)$ given that $Z(\bm{s}- u\bm\tau,t)=x_1$ is not influenced by the spatial dependence structure of the innovation field $W$, and is given by
\begin{equation}\label{eqn:conditional_distribution}
    \mathbb{P} \big(Z(\bm{s},t+u) \leq z_2\ |\ Z(\bm{s}-u \bm{\tau} ,t)=z_1\big) = \mathbb{I}(z_2 \geq a^u z_1)\,\mathrm{exp}\left({-\frac{1-a^u}{z_2}}\right).
\end{equation}
An immediate consequence of~\eqref{eqn:conditional_distribution} is that
\begin{equation}
    \mathbb{P} \big(Z(\bm{s},t+u) = a^uz\ |\ Z(\bm{s}-u \bm{\tau} ,t)=z\big) = \mathrm{exp}\left({-\frac{a^{-u}-1}{z}}\right) > 0.
    \label{Eq_ExplicitDirac}
\end{equation} 
Intuitively, this non-zero conditional probability of $Z(\bm{s},t+u) = a^uz$ arises from taking the maximum of a deterministic and a random term in~\eqref{eqn:max_auto}. 
We see that the conditional distribution contains a mass at $a^u z$, and so the pairwise distribution is a mixture of a Dirac distribution and an absolutely continuous distribution. For two space-time points chosen such that the spatial lag is $\bm \tau$ times the temporal lag $u$, the distribution of the pair $(Z(\bm{s},t+u), Z(\bm{s}-u \bm{\tau} ,t))^{\prime}$ has a mass in $(a^u z, z)$ for any $z>0$.

We end this subsection by noting that in the specific case where the innovation field $W$ is taken to be the spatial Brown--Resnick model on $\mathbb{R}^2$ (see Section~\ref{subsect_remindermaxstable}), the exponent measure of $Z$ in~\eqref{eqn:bivariate_exponent_measure_general} becomes
\begin{align}\label{eqn:bivariate_exponent_measure_BR}
    V_{Z;\bm{h},u}(z_1,z_2) = & \frac{1}{z_1}\Phi\bigg(\frac{\log\big(z_2/(a^uz_1)\big)}{\sqrt{2\gamma(\bm{h}-u\bm{\tau})}} + \sqrt{\frac{\gamma(\bm{h}-u\bm{\tau})}{2}}\bigg)\\
    & + \frac{a^u}{z_2}\Phi\bigg(\frac{\log(a^uz_1/z_2)}{\sqrt{2\gamma(\bm{h}-u\bm{\tau})}} + \sqrt{\frac{\gamma(\bm{h}-u\bm{\tau})}{2}}\bigg)\nonumber +\frac{1-a^u}{z_2},\nonumber
\end{align}
for $\bm{h}-u \bm\tau\neq \bm{0}$ and is given by \eqref{eqn:bivariate_exponent_measure_general_degeneratecase} otherwise. 
 
\subsection{Forecasting strategy}\label{Sec_ModelForecasting_ForecastStrategy}

Let us consider a specific site $\bm{s} \in \mathbb{R}^2$ and a specific time point $t \in \mathbb{N}$, and assume that we aim at forecasting using our model the considered variable at $\bm{s}$ at time $t+u$, for some $u \in \mathbb{N}^+$ (e.g., $u=1$). In other words, we wish to explicit the conditional distribution of $Z(\bm{s},t + u)$ given all observations of the field $Z$ at time $t$. Our forecasting strategy is based on the recurrence \eqref{Eq_RecurrenceLagu}, that can be rewritten
\begin{equation}
\label{Eq_RecurrenceLagu_New}
    Z(\bm{s}, t+u) = \max\big\{ a^u Z(\bm{s}-u\bm{\tau}, t),\ (1-a^u)\widetilde W^{t+u}_{t}(\bm{s})\big\}, \quad \bm{s} \in \Mbb{R}^2.
\end{equation}
In that case, our model allows for exact sampling from the forecasting distribution. The problem is more complex if the realization of $Z(\bm{s}-u\bm \tau,t)$ was not observed (i.e., $\bm{s}-u\bm \tau$ does not lie on our spatial grid), and we propose the following strategy, assuming that it is possible to perform conditional simulation of the spatial max-stable field $W$ in \eqref{Eq_InnovationField}. Using an algorithm for conditional simulation of max-stable fields (e.g., the one by \citet{dombry2013conditional}), one can simulate $N$ realizations, denoted by $y_1, \ldots, y_N$, from the conditional distribution of the random variable $Z(\bm{s}-u\bm \tau,t)$ given all, or a subset of, the available observations of the space-time field $Z$ at time $t$. Then, for any $i=1, \ldots, N$, we simulate a realization $w_i$ of $\widetilde W^{t+u}_{t}(\bm{s}) \stackrel{\mathrm d}= W(\bm s)$ by drawing from a standard Fréchet distribution, and then compute $z_i=\max\{a^u y_i, (1-a^u) w_i\}$ according to \eqref{Eq_RecurrenceLagu_New}; equivalently, the $z_i$ have been drawn from the conditional distribution in \eqref{eqn:conditional_distribution}. The $z_i$ form a random sample from the conditional distribution we are focusing on. 
In the following, we will take the spatial Brown--Resnick field for $W$ owing to its flexibility and suitability for environmental data. 

\section{Inference}
\label{Sec_InferenceSimulations}

This section presents our estimation procedure, which relies on pairwise likelihood techniques due to intractability of the full likelihood as mentioned in Section \ref{subsect_remindermaxstable}. We now take as innovation field $W$ the spatial Brown--Resnick field as it will be
the one we consider in the case study. It is defined by 
\begin{equation}
\label{Eq_Spatial_BR_Field}
W(\bm{s})=\bigvee_{i=1}^{\infty }U_{i}Y_{i}(\bm{s}),\quad \bm{s}\in \mathbb{R%
}^{2},
\end{equation}
where the $(U_i)_{i \geq 1}$ are the points of a Poisson point process on $(0, \infty)$ with intensity function $u^{-2} \Mr{d}u$ and, independently of the $U_i$, the $Y_{i}$ are independent replications of $Y(\bm{s})=\exp \left\{
\epsilon (\bm{s})-\mathrm{Var}(\epsilon (\bm{s}))/2\right\} $ and $%
\left\{ \epsilon (\bm{s})\right\}_{\bm{s}\in \mathbb{R}^{2}} $ is a centered
Gaussian random field with stationary increments and semivariogram $\gamma\left( \bm{h}\right) =\left( \left\Vert \bm{h}%
\right\Vert /\kappa \right) ^{2H}$ with $\bm{h}\in \mathbb{R}^{2}$ and $%
\bm{\theta }=\left( \kappa ,H\right)^{\prime} \in (0,\infty )\times \left(
0,1\right)$.

The true parameter vector is denoted as ${\bm{\psi }}^{\star }=(%
\bm{\theta }^{\star },\bm{\tau }^{\star },a^{\star })^{\prime}$ and is
assumed to belong to a compact set $\Psi \subset \R^{+}\times \left(
0,1\right) \times \R^{2}\backslash \{\bm{0}\}\times \left( 0,1\right)$. We assume that $\{Z(\bm{s},t)\}_{\bm{s}%
\in \R^{2},t\in \N}$ is sampled at locations that lie on a regular
two-dimensional grid with mesh distance $\mu >0$ 
\begin{equation}\label{eqn:Sm}
\mathcal{S}_{m}=\{(\mu \times i_{1},\mu \times i_{2}):(i_{1},i_{2})\in \N^{2}:1\leq
i_{1},i_{2}\leq m\},
\end{equation}
and at $T$ equidistant time points, $t_{i}=i$ for $i=1,\ldots,T$. Further,
for $r\geq 1$, denote by
\begin{equation}\label{eqn:full_design_mask}
\mathcal{H}_{r}=\{\bm{s}=\mu \bm{z}:\bm{z}\in \Z^{2}:||\bm{z}||\leq r\}
\end{equation}
the set of spatial lags between space-time pairs used in the estimation procedure.
Then for some fixed $r\geq 1$ and $p\in \N^{+}$, the pairwise log-likelihood is, for $\bm{%
\tau }\notin \{\bm{h}/u : \bm{h} \in \mathcal{H}_r, u=1, \ldots, p\}$,
\begin{equation}
\mathrm{PL}^{(m,T)}({\bm{\psi }})=\sum_{\bm{s}\in \mathcal{S}%
_{m}}\sum_{t=1}^{T}\sum_{\substack{ \bm{h}\in \mathcal{H}_{r}  \\ 
\bm{s}+\bm{h}\in \mathcal{S}_{m}}}\sum_{\substack{ u=1  \\ t+u\leq T
}}^{p}\log f_{\bm{h},u}\big(Z(\bm{s},t),Z(\bm{s}+\bm{h},t+u);%
{\bm{\psi }}\big),  \label{eqn:pl-likelihoood}
\end{equation}%
where $f_{\bm{h},u}(\cdot,\cdot;{\bm{\psi }})$ is the bivariate
density of $\left( Z(\bm{s},t),Z(\bm{s}+\bm{h},t+u)\right)^{\prime}$ with detailed expression given in Section~\ref{sec:density_function}. If $\bm{%
\tau }\in \{\bm{h}/u : \bm{h} \in \mathcal{H}_r, u=1, \ldots, p\}$,
the distribution of $\left( Z(\bm{s},t),Z(\bm{s}+\bm{h},t+u)\right)^{\prime} $ has an additional
Dirac component as shown in~\eqref{Eq_ExplicitDirac}, in which case the pairwise log-likelihood is undefined for some values of the parameter $a$.

Without loss of generality we assume that $\bm{\tau }^{\star }\neq 
\bm{h}/u$ for all $\bm{h}\in \mathcal{H}_{r}$ and $u\in \{1,\ldots,p\}$, which is not too restrictive as there is no reason for the grid of sites to be related to $\bm{
\tau }^{\star }$. We propose in Section~\ref{App_Diagnos_ratio} a diagnostic to check this assumption using the ratio random field $\chi$ defined in \eqref{eqn:chi_def}. It is thus unnecessary to compute $\mathrm{PL
}^{(m,T)}({\bm{\psi }})$ for ${\bm{\psi }}$ such that $
\bm{\tau }=\bm{h}/u$ for some $\bm{h}$ and $u$. We therefore define the pairwise log-likelihood estimator of 
${\bm{\psi }}^{\star }$ by
\begin{equation}\label{eqn:pl-estimator}
{\bm{\hat{\psi}}}=\arg \max_{{\bm{\psi \in }}\Psi
_{\varepsilon }}\mathrm{PL}^{(m,T)}({\bm{\psi }}),
\end{equation}
where 
\begin{eqnarray}\label{eqn:psi_epsilon}
\Psi_{\varepsilon } &=&[\varepsilon ,\varepsilon ^{-1}]\times \lbrack
\varepsilon ,1-\varepsilon ]\times \lbrack -\varepsilon^{-1} ,\varepsilon
^{-1}]^{2}\times \lbrack \varepsilon ,1-\varepsilon ] \\
&&\cap \left\{ {\bm{\psi } \in \R^5}: ||\bm{\tau }-\bm{h}
/u|| \geq \varepsilon \text{ for all } \bm{h}\in \mathcal{H}_{r},\ 
u = 1,\ldots,p\right\},\nonumber
\end{eqnarray}
with $0<\varepsilon <\min \{1/2,\mu /p\}$ so that $\Psi _{\varepsilon }$ is not empty. Throughout this section,
when optimizing functions with respect to a subset of the components of $\bm \psi$, the search space is assumed to be the associated projection of $\Psi_{\varepsilon}$.   We show in Section~\ref{sec:pwle} that, if ${\bm{\psi }}^{\star } = (\kappa^\star, H^\star, \bm\tau^\star, a^\star)^{\prime}\in
\Psi _{\varepsilon }$, then ${
\bm{\hat{\psi}}}$ is almost surely consistent (see Theorem \ref{thm:consistency}) and asymptotically
normal (see Theorem \ref{thm:normality}) as the number of spatial and temporal observations increase to infinity (i.e., $m,T\to \infty $). These results mainly stem from the mixing properties of our space-time max-stable model.

In practice we estimate ${\bm{\psi }}^{\star }$ similarly as in \cite{embrechts2016space}. 
Let us denote the observed data by $\{Z_{\bm{s}, t}\}_{(\bm{s},t) \in \mathcal{S}_m \times\{1, \ldots, T\}}$. As a first step, the estimation
of $\bm{\theta}^{\star}$ is carried out by maximizing the spatial pairwise log-likelihood \citep[see][Section 3.2]{padoan2010likelihood} defined by
\begin{equation}
\mathrm{PL}_{\mathrm{S}}^{(m,T)}({\bm{\theta }})=\sum_{\bm{s}\in \mathcal{S}%
_{m}}\sum_{t=1}^{T}\sum_{\substack{ \bm{h}\in \mathcal{H}_{r}  \\ 
\bm{s}+\bm{h}\in \mathcal{S}_{m}}} \log f_{\bm{h},0}\big(Z_{\bm{s},t},Z_{\bm{s}+\bm{h},t};%
{\bm{\theta }}\big), \label{eqn:pl-likelihoood_spatial}
\end{equation}%
where the parameters $a$ and $\bm \tau$ do not appear in the expression of $f_{\bm{h},0}$.
Once $\bm{\theta}^{\star}$ is known, it is held fixed and we estimate $a^{\star}$ and $\bm \tau^{\star}$ by maximizing \eqref{eqn:pl-likelihoood} with respect to $a$ and $\bm \tau$. In that second step, consistently with our assumption, we exclude from the optimization procedure the values of $\bm \tau$ within a distance $\varepsilon$ of the set $\{\bm{h}/u: \bm{h} \in \mathcal{H}_r, u=1, \ldots, p\}$. The robustness with respect to the choice of $\varepsilon$ should be assessed.

Finally, to derive confidence bounds for the maximum pairwise log-likelihood estimator of ${\bm{\psi }}^{\star }$, we employ the following non-parametric bootstrap procedure which only involves the terms of the
pairwise log-likelihood function and does not require creating new datasets based on rearrangements (except when accounting for the marginal uncertainty). We take many (e.g., 100) bootstrap samples of the set of time points $\{1,\ldots,T\}$, and, for each bootstrap sample $\mathcal{B}$, we estimate the parameters of the model as follows. First, the margin at each grid point is transformed to a standard Fr\'echet distribution by fitting a GEV distribution to the data at times in $\mathcal{B}$. The spatial parameters' estimates $\hat \kappa$ and $\hat H$ are then obtained by computing
\begin{equation*}
        (\hat \kappa, \hat H) = \arg \max_{\kappa, H}
        \sum_{\bm{s}\in \mathcal{S}_m}
        \sum_{t\in \mathcal{B}}
        \sum_{\substack{\bm{h}\in\mathcal{H}_r \\
              \bm{s} + \bm{h} \in \mathcal{S}_m}}
        \log f_{\bm{h},0}\big(Z^{\mathcal{B}}_{\bm{s},t}, Z^{\mathcal{B}}_{\bm{s} + \bm{h},t};\kappa,H\big),
\end{equation*}
where $\{Z^{\mathcal{B}}_{\bm s,t}\}_{(\bm{s},t) \in \mathcal{S}_m \times\{1, \ldots, T\}}$ is the transformed dataset.
Secondly, to estimate the temporal parameters $\bm \tau^{\star}$ and $a^{\star}$, we transform the entire original dataset (without bootstrapping) to have standard Fr\'echet margins by fitting a GEV distribution to the original time series at each grid point. Using this transformed dataset $\{\tilde{Z}_{\bm s, t}\}_{(\bm{s},t) \in \mathcal{S}_m \times\{1, \ldots, T\}}$, we obtain the temporal parameters' estimates $\hat{\bm\tau}$ and $\hat{a}$ by computing
\begin{equation}
\label{Eq_BoostrapSecondStep}
        (\hat {\bm \tau}, \hat a) = \arg \max_{\bm \tau, a}
        \sum_{\bm{s}\in S_m}
        \sum_{t\in \mathcal{B}}
        \sum_{\substack{\bm{h}\in \mathcal{H}_{r} \\
            \bm{s}+\bm{h}\in \mathcal{S}_{m}}}
        \sum_{\substack{u = 1 \\ t+u\leq T}}^p
        \log f_{\bm{h},u}\big(\tilde{Z}_{\bm s, t}, \tilde{Z}_{\bm s + \bm h, t+u};\hat\kappa,\hat H, \bm\tau, a\big).
\end{equation}
The necessity of fitting the GEV distribution to the entire time series stems from~\eqref{Eq_BoostrapSecondStep}, which involves data at $t+u$ for $t \in \mathcal{B}$ although $t+u$ may not belong to $\mathcal{B}$. Thus, when fitting the GEV distribution to data associated with times in $\mathcal{B}$ only, the obtained parameters are incompatible with some data points at some grid points, which may create undefined values in the resulting transformed dataset. It is for this reason that $\tilde{Z}$ is constructed from all available data points.

We bootstrap the terms in the pairwise log-likelihood rather than the observations themselves. Compared to the well-known block-bootstrap \citep{kunch1989}, this has the advantage to fully preserve dependence in both space and time by avoiding the decomposition into blocks and their rearrangement. No arbitrary choice of block size is needed and no points are privileged or under-represented since there are no block boundaries. 

\medskip

We conclude this section with a discussion of the inference method for the DKS model by \cite{davis2013statistical}. The pairwise log-likelihood for their model is also given by~\eqref{eqn:pl-likelihoood}, where the appropriate bivariate densities are used. The authors restricted the design mask $\mathcal{H}_r$ to vectors with non-negative integer components, and further excluded the $\bm{0}$ vector. We choose to keep these vectors to calibrate their model, as justified in Section~\ref{sec:design_mask_app}.

\section{Case study}
\label{Sec_CaseStudy}

Let us now return to the hourly data presented in Section~\ref{Sec_Data_Data}. For each of the 216 grid points, we model the margin in space at each grid point $\bm{s}_i$, $i=1, \ldots, D$, by a GEV distribution with location, scale and shape parameters $\mu_{\bm{s}_i}$ , $\sigma_{\bm{s}_i}$ and $\xi_{\bm{s}_i}$, fitted to the 105 temporal observations using maximum likelihood. 

Figure~\ref{fig:gev_params} in Section~\ref{sec:single_site_params} shows that the location and scale parameters are higher towards north-west, i.e., closer to the sea, consistent with physical intuition. The estimated GEV parameters are used to transform the spatial margins to the standard Fr\'echet distribution. This procedure, as opposed to first modeling these parameters using trend surfaces, maintains the most accurate view of the spatio-temporal dependence structure since the distribution of the observations at each grid point is well-approximated by the standard Fr\'echet distribution.

Throughout we use on the standardized dataset the model presented in Section \ref{Subsec_Model} with as innovation $W$ the Brown--Resnick field \eqref{Eq_Spatial_BR_Field} associated with the semivariogram $\gamma(\bm h) = (\| \bm  h \|/\kappa^\star)^{2H^\star}$ for some parameter $\bm\psi^\star = (\kappa^\star, H^\star, \bm\tau^\star, a^\star)^{\prime}$, where $\bm \tau^\star = (\tau_1^\star,\tau_2^\star)^{\prime}$ and $a^\star$ are the advection and decay parameters, respectively. The induced temporal stationarity is appropriate since the considered dataset involves a relatively short time window (of the order of four days), i.e., corresponding to a single meteorological event.  

We assess the various goodness-of-fits in-sample, rather than out-of-sample on a validation set randomly subsampled from our data. This is imposed by our forecasting procedure which requires the observations at all grid points at the time the forecast is performed, thereby excluding the possibility of removing individual data points. However, this appears to be suitable owing to the parsimony of our model (involving only five scalar parameters) which leads to rather low overfitting risks. An effective alternative validation strategy would involve testing our model on a comparable period in terms of synoptic weather situation or weather regime (see Section~\ref{Sec_Discussion} and Section~\ref{App_Operational} for more details about weather regimes), but we believe that this is out of the scope of this work.


\subsection{Calibration to data}\label{Sec_CaseStudy_Performance}

We follow the estimation procedure outlined in Section~\ref{Sec_InferenceSimulations}, which requires the absence of pairs $\bm h\in \mathcal{H}_r$ and $u\in {1,\ldots,p}$ such that $\bm\tau^\star = \bm h/u$. Our diagnostic analysis (Figure~\ref{Fig_DataCurves} in Section~\ref{App_Diagnos}) reveals no violations of this assumption in the observed data.
 
In the first step we estimated $\bm{\theta^\star}$ by maximizing \eqref{eqn:pl-likelihoood_spatial} with $r=21$ (to account for all pairs in space). In the second one, we estimated $\bm \tau^\star$ and $a^\star$ by maximizing \eqref{eqn:pl-likelihoood} with the same value of $r$ as previously and $p=1$ in order to avoid using too many pairs in time as explained in \citet[][Section 7]{davis2013b}. The huge number of space-time pairs ($215 \times 216 / 2 \times 105 = 2\,438\,100$ in the first step and $216^2 \times 104 = 4\,852\,224$ in the second one) allows us to use our theoretical results about the asymptotic behavior of the maximum pairwise likelihood estimator to guarantee the accuracy of our estimates. The symmetric 95\% confidence bounds were derived using the bootstrap procedure expounded in Section~\ref{Sec_InferenceSimulations}.

Table~\ref{tab:estimates} shows that the estimate of $\kappa^\star$ is quite large relative to the size of the domain, indicating a large-scale (synoptic) system, and that the estimated smoothness parameter $2 H^\star$ is quite typical for wind gust data. The estimate of the advection parameter $\bm \tau^\star = (\tau_1^\star,\tau_2^\star)^{\prime}$ indicates a general movement of the spatial patterns towards the east ($\tau_1^\star > 0$), with a small component of the velocity in the southern direction ($\tau_2^\star < 0$). We may thus interpret this phenomenon as being driven by west-northwest winds, which is consistent with what has been observed in Figure~\ref{Fig_Snapshots}. The rate of decay of temporal dependence $a^\star$ is estimated to be close to unity, which implies a slow decay of dependence in time that is consistent with large-scale weather features being persistent over several hours.
\begin{table}[]
    \centering
    \caption{Estimated model parameters using pairwise likelihood and associated bootstrap $95\%$ confidence interval (CI).}
    \begin{tabular}{c|ccccc}\label{tab:estimates}
         &  $\kappa^{\star}$ & $2H^{\star}$ & $\tau_1^{\star}$ & $\tau_2^{\star}$ & $a^{\star}$ \\
         \hline
         Estimate & 2.19 & 1.33 & 0.35 & -0.14 & 0.97\\
         CI & (1.83 -- 2.51) & (1.24 -- 1.41) & (0.30 -- 0.38) & ($-$0.17 -- $-$0.11) & (0.94 -- 0.99)
    \end{tabular}
    \label{tab:my_label}
\end{table}


The fit of our model to the data using the parameters in Table~\ref{tab:estimates} is evaluated by three strategies that we outline in the remainder of this section. The first concerns the marginal and spatial features of our model. The second is a comparison of the cross-correlations of the model with those of the data. In the third, we use the model to forecast wind gust speeds at later time steps, and compute a score for our predictions based on what was actually observed.

\subsection{Single-site marginal and spatial goodness-of-fits}

In this section we assess the single-site marginal and spatial performance of our model fitted to the data. The left panel of Figure~\ref{Fig_GoodnessFitBR} shows that the GEV distribution fitted to the data at the (randomly) chosen grid point matches the empirical distribution, and the right one indicates that the proposed model fits the pairwise extremal dependence structure of the data reasonably well, suggesting that the spatial Brown--Resnick random field (as mentioned in Section~\ref{Subsec_Model}, for any $t \in \mathbb{N}$, the spatial field $\{Z(\bm{s}, t)\}_{\bm{s} \in \mathbb{R}^2}$ is max-stable with the same distribution as that of $W$, i.e., a spatial Brown--Resnick field) is a fairly good model for the spatial dependence in our data. 
\begin{figure}[h]
    \centering
    \begin{subfigure}[b]{0.4\textwidth}
        \centering
\includegraphics[width=\textwidth]{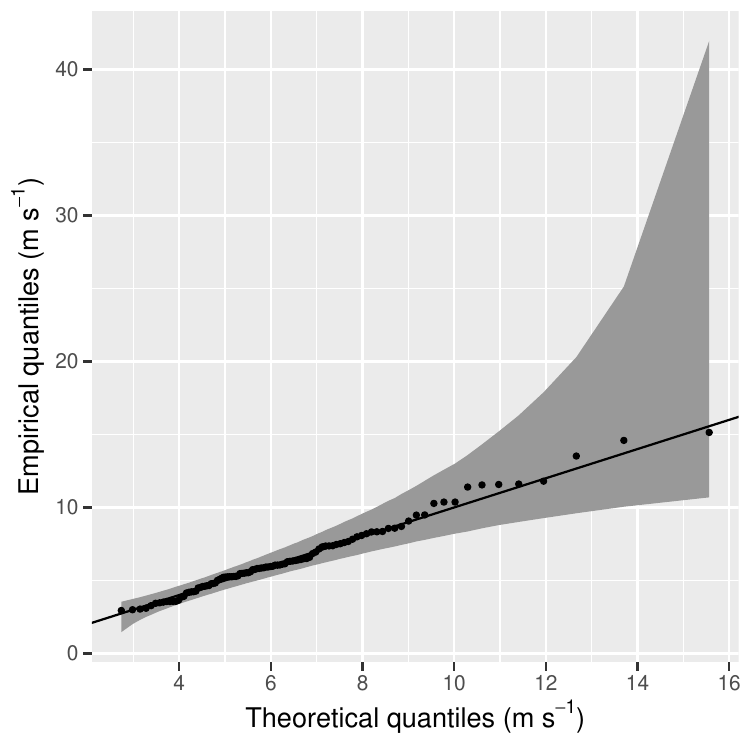}
        \end{subfigure}
        \begin{subfigure}[b]{0.4\textwidth}  
            \centering 
\includegraphics[width=\textwidth]{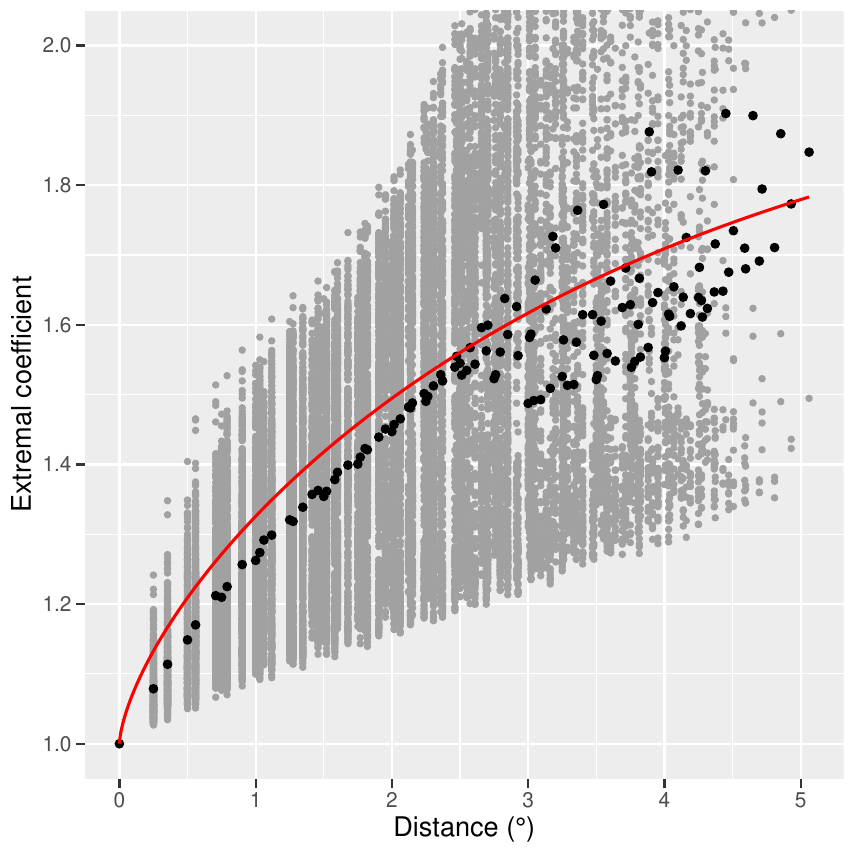}
        \end{subfigure}
        \caption{Validation of our model in terms of single-site marginal distributions and spatial dependence structure. On the left, quantile-quantile plot of the fitted GEV distribution for the grid point with coordinates $2.5^{\circ}$ longitude and $47^{\circ}$ latitude. On the right, theoretical spatial pairwise extremal coefficient function from the fitted Brown–Resnick model (red line) and empirical spatial pairwise extremal coefficients (dots). The grey and black dots are pairwise and binned estimates, respectively. The empirical extremal coefficients have been computed based on the empirical F-madogram using the obtained GEV parameters. The binned estimates have been obtained by first averaging, for any distance, the F-madogram estimates over all pairs of grid points at that distance.}
\label{Fig_GoodnessFitBR}
\end{figure}

\subsection{Goodness-of-fit of cross-correlations}

In order to also assess the temporal dynamics of our model, we now compare the associated cross-correlations with those observed in the data. Recall that after a marginal transformation, the resulting dataset $\{Z_{\bm{s}, t}\}_{(\bm{s},t) \in \mathcal{S}_m \times\{1, \ldots, T\}}$ has margins that are approximately standard Fréchet, a distribution that does not have finite second-order moments. To remedy this issue, we consider the logarithm of our data, i.e., $\{\log Z_{\bm{s}, t}\}_{(\bm{s},t) \in \mathcal{S}_m \times\{1, \ldots, T\}}$, which have approximately standard Gumbel margins and thus finite second-order moments.

By stationarity, the cross-correlations between two observations of $\log Z$ depend only on the space-time lag $(\bm h, u)$, where $\bm{h}\in \R^2$ and $u\in \N$. Thus, we define
$$\rho_{\bm h, u} = \mathrm{Corr}\big(\log Z(\bm 0, 0), \log Z(\bm h, u)\big),$$
where $Z$ refers to our model or the DKS one depending on the context. 

For each of the space-time lags $(\bm{h},u)$ considered in Figure~\ref{fig:CrossCorrs}, we compute an empirical cross-correlation coefficient $\bar\rho_{\bm{h},u}$, which is the average over the set
\begin{equation}\label{eqn:set_of_rho_hat}
    \{\hat\rho_{\bm{s},\bm{h},u}: \bm{s}, \bm{s}+\bm{h}\in S_m\},
\end{equation}
where
\begin{equation}\label{eqn:correlation_estimator_def}
\hat\rho_{\bm{s},\bm{h},u} = \frac{6}{(n-u)\pi^2} \sum_{t=1}^{n-u} \left( \log Z_{\bm{s}, t} - \hat{\mu}_{\bm{s}} \right)\left( \log Z_{\bm{s} + \bm{h}, t+u} - \hat{\mu}_{\bm{s}+\bm{h}} \right),
\end{equation}
and
$$ \hat{\mu}_{\bm{s}} = \frac{1}{n} \sum_{t=1}^{n} \log Z_{\bm{s}, t}, \qquad \bm s\in S_m.$$
The factor of $6/\pi^2$ appearing in~\eqref{eqn:correlation_estimator_def} corresponds to the inverse of the variance of the Gumbel distribution.

The theoretical cross-correlations were computed using numerical integration based on Hoeffding's lemma which states, for any random variables $X, Y$ with finite second-order moments, that
\begin{equation}
\label{Eq_Hoeffding_Lemma}
    \mathrm{Cov}(X,Y) = \int_{-\infty}^\infty\int_{-\infty}^\infty \left[ \P(X \leq x, Y \leq y) - \P(X \leq x)\P(Y\leq y) \right] \, \mathrm{d}x\mathrm{d}y.
    \end{equation}
The cross-correlation for the spatial lag $\bm h$ and the temporal lag $u$ is obtained by taking $X = \log Z_{\bm s, t}$ and $Y = \log Z_{\bm s + \bm h, t + u}$ in \eqref{Eq_Hoeffding_Lemma}, where $\P(X \leq x, Y \leq y)$ can be deduced by combining \eqref{Eq_Multivariatedf_Maxstable} with the exponent measure~\eqref{eqn:bivariate_exponent_measure_BR} and using the model parameters in Table~\ref{tab:estimates}. To compute the cross-correlations for the DKS model, we use the estimates
$$(\hat \kappa_s, \hat \kappa_t, \hat \psi_s, \hat \psi_t)^{\prime} = (6.98, 4.72, 1.82, 1.47)^{\prime},$$
which were obtained using the pairwise likelihood-based strategy outlined in \cite{davis2013statistical}, with $\mathcal{H}_r$ defined as in~\eqref{eqn:full_design_mask} (see also, Section~\ref{sec:design_mask_app}).

\begin{figure}
    \centering
    \includegraphics[width=0.55\linewidth]{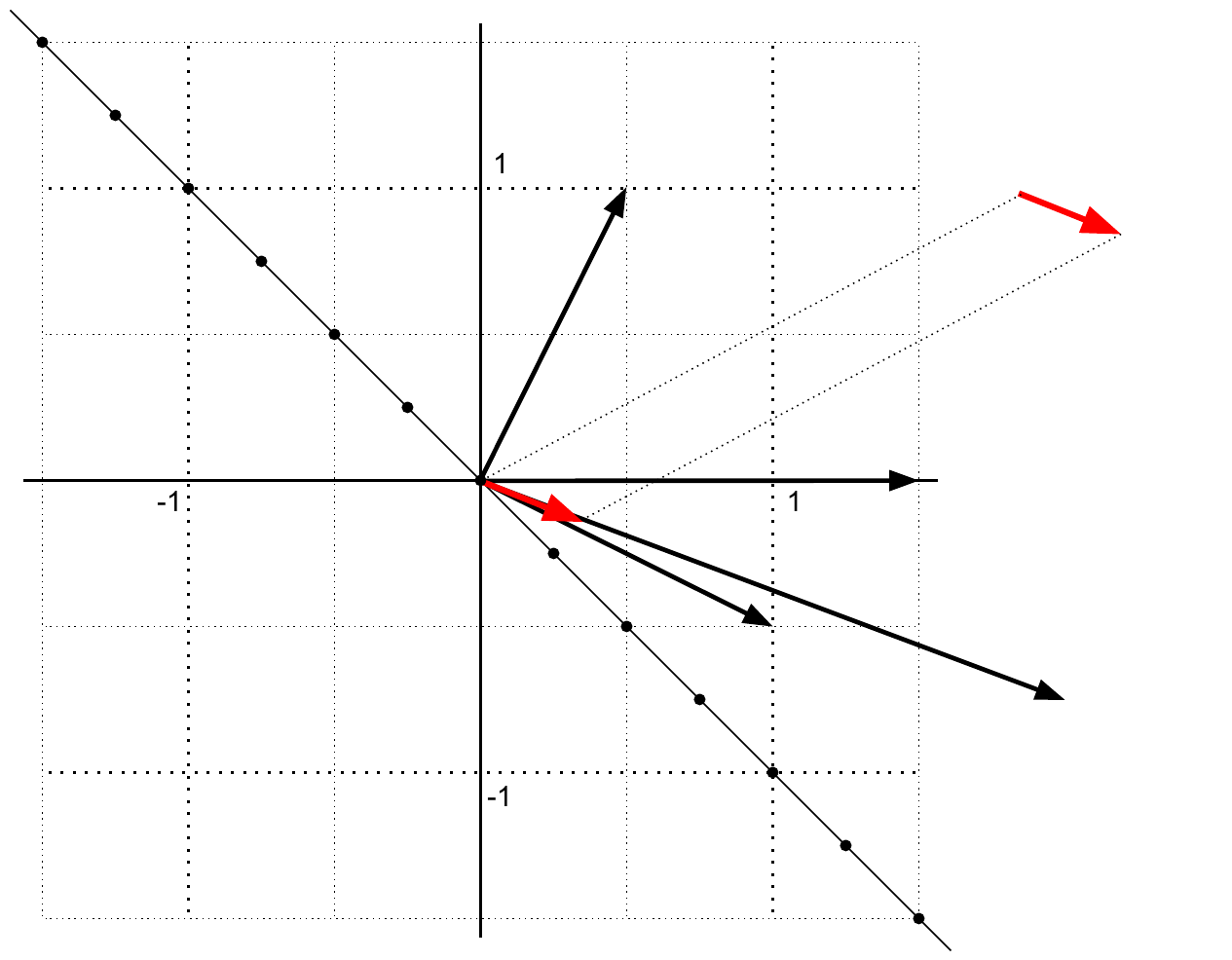}
    \put(-260,200){Figure~\ref{Fig_CorrelStructure1}}
    \put(-70,105){$(1.5,0)$ Figure~\ref{Fig_CrossCorr60}}
    \put(-110,55){$(1,-0.5)$}
    \put(-96,42){Figure~\ref{Fig_CrossCorr4min2}}
    \put(-40, 65){$(2,-0.75)$ Figure~\ref{Fig_CrossCorr8min3}}
    \put(-150, 165){$(0.5,1)$ Figure~\ref{Fig_CrossCorr24}}
    \put(-25, 155){$(\hat{\tau}_1, \hat{\tau}_2)$ Table~\ref{tab:estimates}}
    \caption{Points on the diagonal line represent the spatial lags used in Figure~\ref{Fig_CorrelStructure1}. The various spatial lags used in Figure~\ref{fig:CrossCorrs} are shown as black vectors. Finally, the vector $(\hat{\tau}_1, \hat{\tau}_2)^{\prime}$ obtained from pairwise likelihood estimation is shown in red.}
    \label{fig:vectors}
\end{figure}

\begin{figure}[!h]
\centering 
\includegraphics[scale = 0.6]{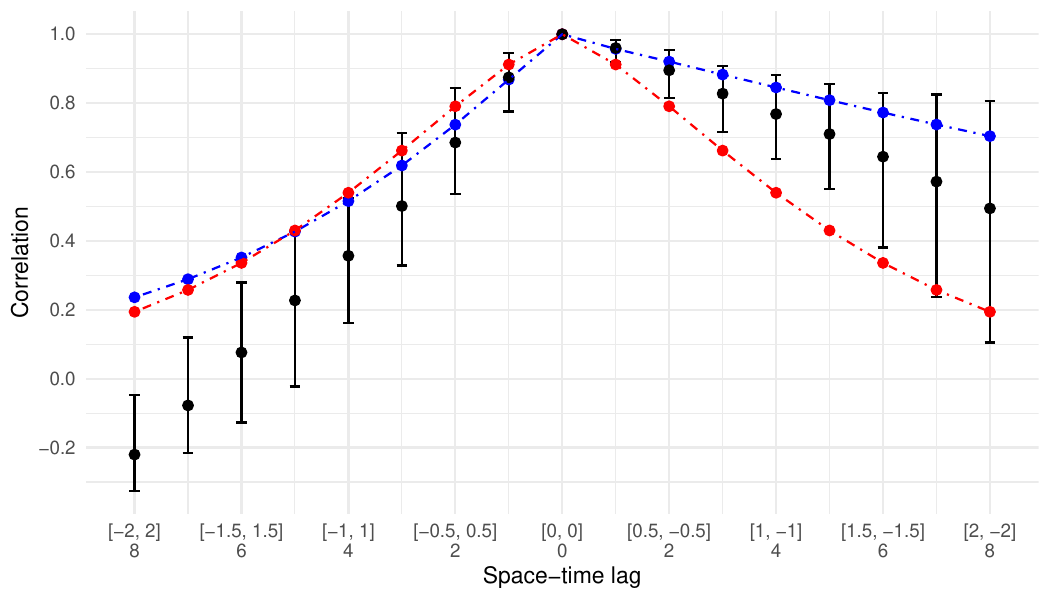} 
\caption{For each of the space-time lags shown on the x-axis (top: $\bm{h}$; bottom: $u$), the empirical cross-correlation $\bar\rho_{{\bm{h},u}}$ is plotted in black. The theoretical cross-correlations of our model fitted to the same data are plotted in blue. Likewise, the cross-correlations of the DKS model are plotted in red. Error bars show the 2.5\% and 97.5\% quantiles of the set in \eqref{eqn:set_of_rho_hat} for the space-time lags $(\bm h, u)$ considered.}
\label{Fig_CorrelStructure1}
\end{figure}

\begin{figure}
    \centering
    \begin{subfigure}{0.49\textwidth}
        \centering
        \includegraphics[width=\textwidth]{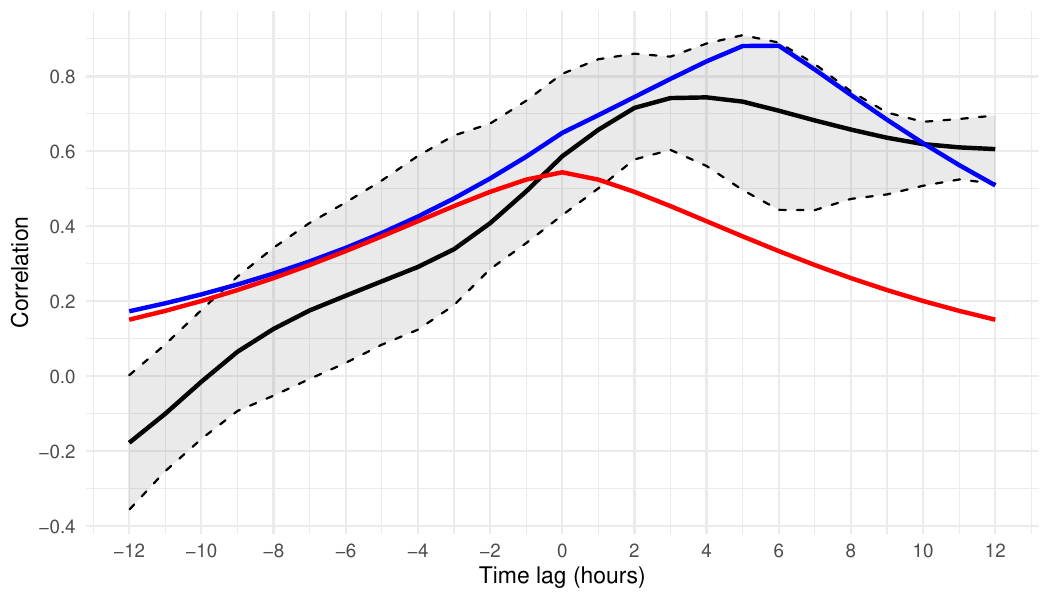}
        \caption{$\bm{h}=(2,-0.75)^{\prime}$}
        \label{Fig_CrossCorr8min3}
    \end{subfigure}
    \begin{subfigure}{0.49\textwidth}
        \centering
        \includegraphics[width=\textwidth]{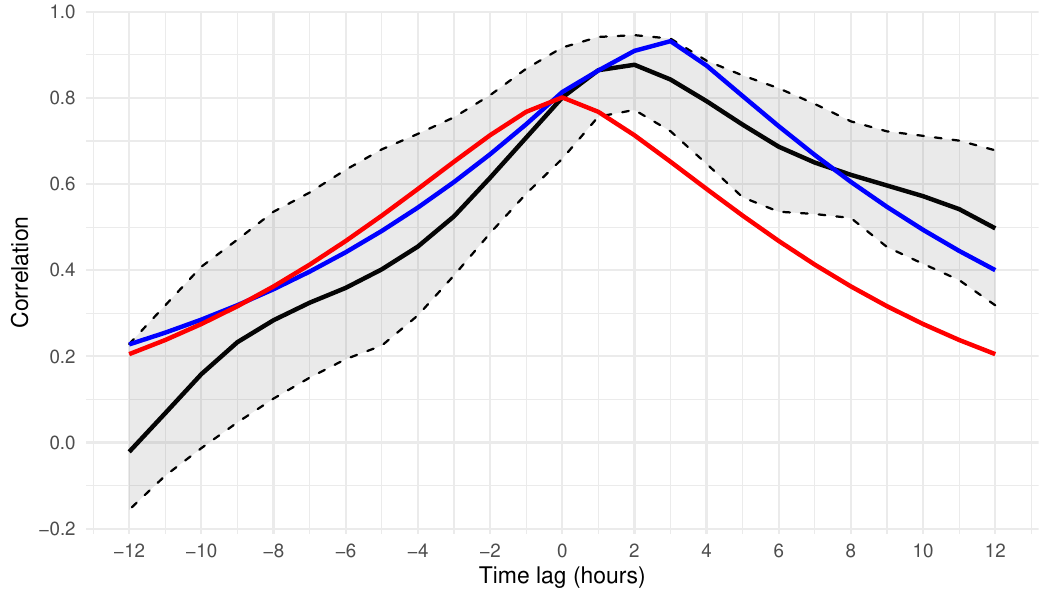}
        \caption{$\bm{h}=(1,-0.5)^{\prime}$}
        \label{Fig_CrossCorr4min2}
    \end{subfigure}
    \begin{subfigure}{0.49\textwidth}
        \centering
        \includegraphics[width=\textwidth]{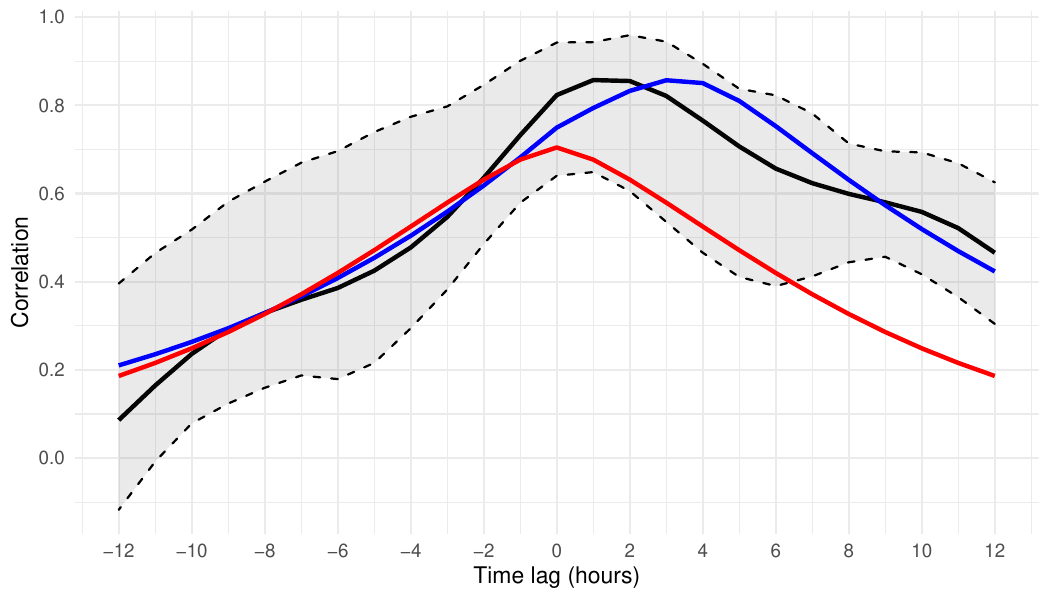}
        \caption{$\bm{h}=(1.5, 0)^{\prime}$}
        \label{Fig_CrossCorr60}
    \end{subfigure}
    \begin{subfigure}{0.49\textwidth}
        \centering
        \includegraphics[width=\textwidth]{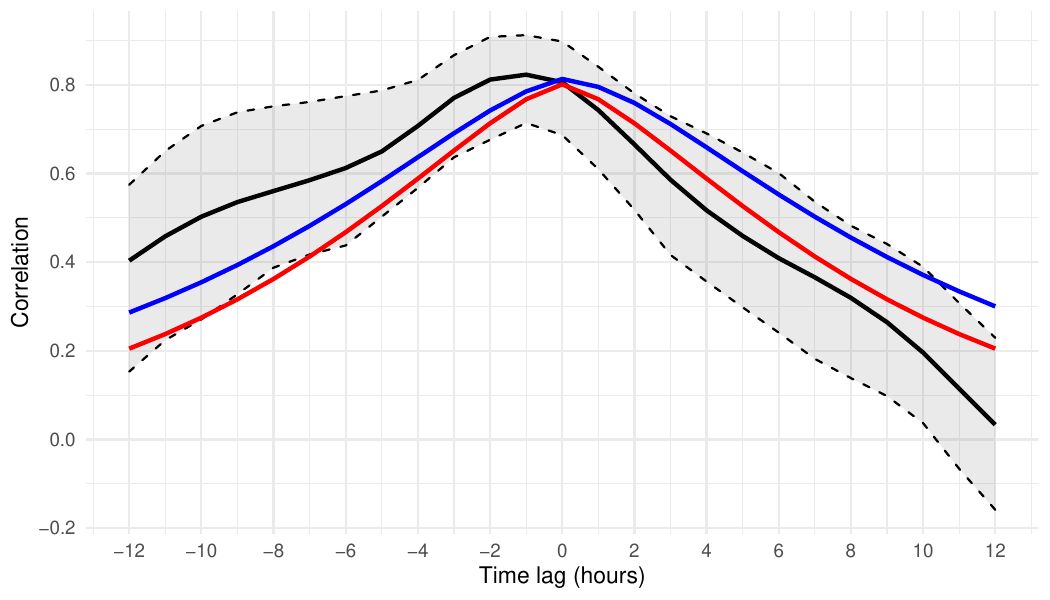}
        \caption{$\bm{h}=(0.5,1)^{\prime}$}
        \label{Fig_CrossCorr24}
    \end{subfigure}
    \caption{For each plot (corresponding to a specific spatial lag $\bm h$),  the empirical cross-correlation $\bar\rho_{{\bm{h},u}}$ is plotted in black for each of the time lags $u$ shown on the x-axis. The theoretical cross-correlations of our model fitted to the same data are plotted in blue. Likewise, the cross-correlations of the DKS model are plotted in red. The dashed lines delimiting the shaded area depict the 2.5\% and 97.5\% quantiles of the set in \eqref{eqn:set_of_rho_hat} for the space-time lags $(\bm h, u)$ considered.}
    \label{fig:CrossCorrs}
\end{figure}


The presence of the $\bm \tau$ parameter in our model allows us to account for the advection and thus for the resulting asymmetric spatio-temporal correlation structure often encountered in atmospheric data. By contrast, the DKS model imposes a symmetry on the correlation structure (specifically, $\rho_{\bm h, u} = \rho_{-\bm h, u}$ and $\rho_{\bm h, u} = \rho_{\bm h, -u}$ for any $\bm h \in \mathbb{R}^2$, $u \in \mathbb{N}^{+}$) which may not adequately capture the dynamics typically observed in atmospheric applications.

To investigate this, in Figure~\ref{Fig_CorrelStructure1}, we choose $(-\bm h, u)$ as $x$-coordinate on the left when $(\bm h, u)$ appears on the right, so that $\rho_{\bm h, u}$ and $\rho_{-\bm h, u}$ can be compared. Moreover, we vary $u$ with the lag $\bm h$ such that $\bm h/u = (0.25, -0.25)^{\prime}$ for all space-time lags on the right side of the plot, the idea being that $\bm h $ is not too far from $ u \bm \tau^\star$, allowing us to track the system's advective motion through time. In the data, we expect the correlation to be largest when $|u|$ is small (i.e., when the points are close in time) and when $\bm h \approx u\bm \tau^\star$ (in the direction of the advection). In Figure~\ref{Fig_CorrelStructure1}, $\bm h/u = (0.25, -0.25)^{\prime}$ and $(-0.25, 0.25)^{\prime}$ on the right and the left, respectively. Although $(0.25, -0.25)^{\prime}$ does not fall into the confidence bounds obtained for $(\tau_1^\star, \tau_2^\star)$ in Table~\ref{tab:my_label} (i.e., we are not exactly following the advection; see Figure~\ref{fig:vectors}), it is much closer to $\bm \tau^\star$ (estimated at $(0.350, -0.139)^{\prime}$) than $(-0.25, 0.25)^{\prime}$ is, explaining why the observed correlation is larger for the lags on the right than for those on the left. This is captured fairly well by our model, but not by the DKS model which shows a symmetric curve associated with underestimated correlations for the lags on the right and overestimated ones for those on the left, in such an extent that theoretical cross-correlations are not contained by most of the 95\% confidence intervals of the observed cross-correlations (see Figure~\ref{Fig_CorrelStructure1}).
On the other hand, owing to our model's ability to capture asymmetry, all related theoretical cross-correlations on the right of Figure~\ref{Fig_CorrelStructure1} are contained in the confidence intervals. A statistical hypothesis test based on these confidence intervals would reject the DKS model but not ours. On the left of the plot, where space-time lags correspond to a direction which is roughly opposite to the true advection, both models have similar cross-correlations that exceed the correlations observed in the data (especially for large space-time lags).

Each plot in Figure~\ref{fig:CrossCorrs} features a symmetry about the time lag $u=0$, such that if $(\bm h, u)$ appears on the right, then $(\bm h, -u)$ appears on the left. As discussed previously, this leads to symmetry in the cross-correlations for the DKS model, where maximum correlation is modeled at $u=0$. The aforementioned figures show that the data exhibit a clear asymmetry which, unlike the DKS model, our model can capture. Indeed, the point of maximal correlation seen in our model can be skewed away from $u=0$. In the case of large dependence in time  ($a\approx 1$), the time lag at which our model exhibits maximal correlation can be shown to be
$u \approx \langle \bm h,\bm \tau^\star \rangle / \|\bm \tau^\star\|^2$, where $\langle \cdot, \cdot\rangle$ denotes the scalar product. Thus, we expect maximal correlation on the right side of the plots if $\langle\bm h,\bm \tau^\star\rangle > 0$, and on the left side if $\langle\bm h,\bm \tau^\star\rangle < 0$. Moreover, by choosing $\bm h$ in the direction of $\hat{\bm \tau}$, the asymmetry in the curve corresponding to our model is expected to increase with $\|\bm h\|$ (see Figure~\ref{Fig_CrossCorr8min3} in comparison to Figure~\ref{Fig_CrossCorr4min2}). 

In Figure~\ref{Fig_CrossCorr4min2}, the spatial lag $\bm h$ is chosen to be nearly in line with $\hat {\bm \tau}$ as in Figure~\ref{Fig_CrossCorr8min3}, but with $\|\bm h\|$ smaller; see Figure~\ref{fig:vectors}. In both cases, the cross-correlations for the DKS model fall outside of the 95\% confidence intervals for positive time lags. Figure~\ref{Fig_CrossCorr60} assesses the model's performance when $\bm h$ is neither in the direction of, nor perpendicular to $\hat {\bm \tau}$. Our model's theoretical cross-correlations remain within the confidence bounds everywhere while those associated with the competing model fail to do so for large positive time lags. In Figure~\ref{Fig_CrossCorr24}, $\bm h$ is chosen to be nearly perpendicular to $\hat{\bm \tau}$. This leads to symmetric cross-correlations for our model, although the data show a peak correlation for $u\approx -1$, which hints that our model does not capture perfectly all features of the data. Nevertheless, the cross-correlations of both models fall within the confidence intervals for most of the chosen space-time lags.

These diagnostic plots demonstrate that, over a large range of space-time lags, our model's cross-correlations are much more compatible with the observed ones than the DKS's ones are. This is especially so when the spatial lag tends to align with the storm's advection direction.

\subsection{Forecasting skill}\label{Sec_CaseStudy_Forecasting}

In addition to comparing the cross-correlations of the models with those of the data, we use the strategy described in Section~\ref{Sec_ModelForecasting_ForecastStrategy} to generate forecasts from our model, assess its forecasting skill and compare it to that of the DKS model.

Our aim is to forecast wind speeds at the space-time point $(\bm s_0, t_0+u)$ based on gridded observations taken at or before time $t_0$, where $u$ represents the forecast horizon (lead time). For our model, if there are no data at $(\bm s_0-u\hat{\bm \tau}, t_0)$, we simulate the data at this site conditioned on the four sites that define the vertices of the cell containing $\bm s_0-u\hat{\bm \tau}$ (see Figure~\ref{Fig_s_min_utau}). Empirical evidence suggests that including other sites has a negligible impact on the distribution of the conditional simulation (not shown).

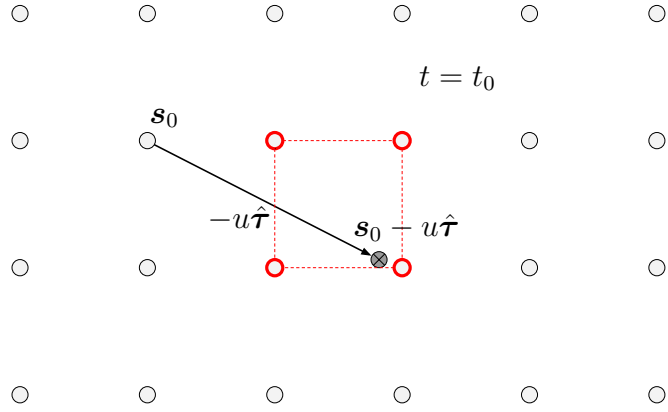
\begin{figure}
    \centering
    \includegraphics[scale = 0.4]{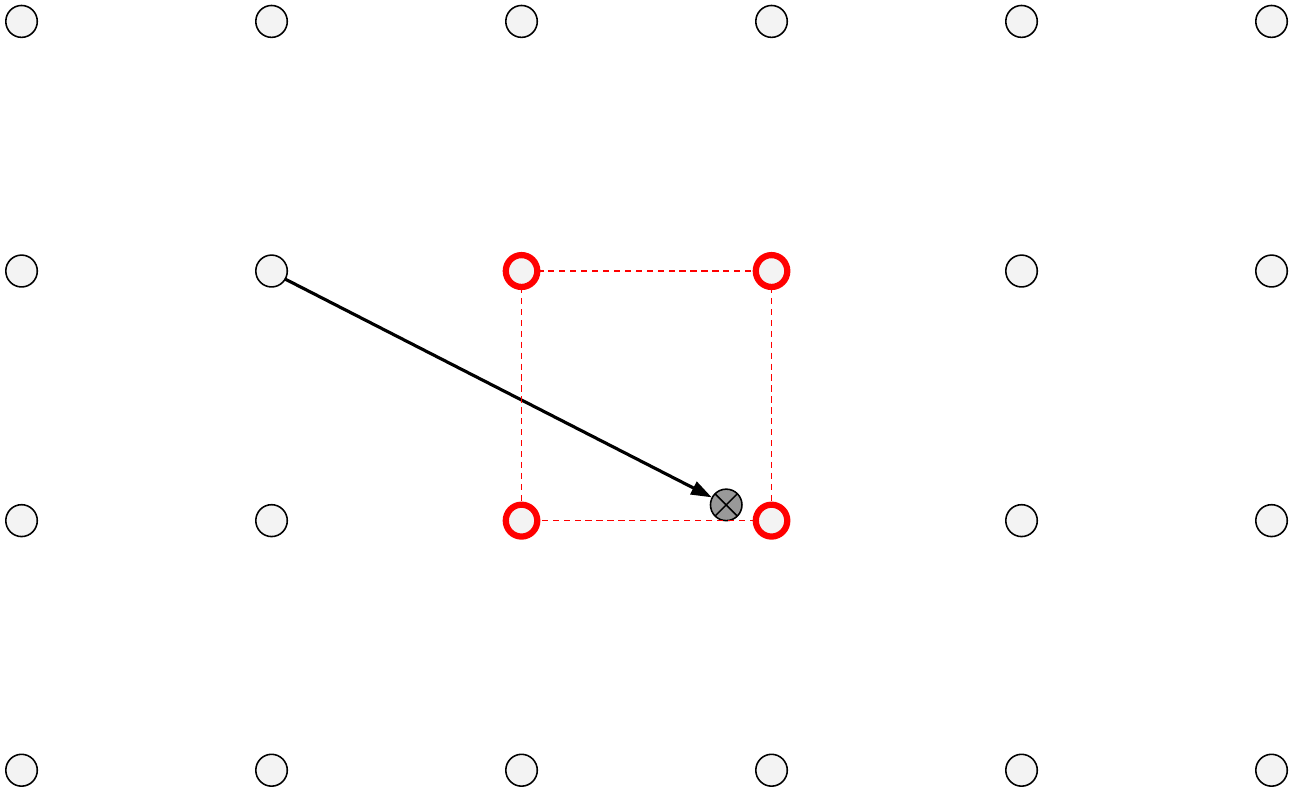}
    \put(-120,67){$\bm s_0-u\hat{\bm \tau}$}
    \put(-175,70){$-u\hat{\bm{\tau}}$}
    \put(-197,110){$\bm s_0$}
    \put(-95,124){$t=t_0$}
    \caption{In red, the four sites that are used to conditionally simulate the random field at $\bm s_0-u\hat{\bm \tau}$. Time is fixed at $t=t_0$.}
    \label{Fig_s_min_utau}
\end{figure}

In order to forecast using the DKS model, an exact approach would be to use a conditional simulation of a three-dimensional Brown--Resnick random field (with two spatial dimensions and one temporal dimension), but currently available softwares do not allow that. Thus, when performing forecast at any site $\bm s_{0} \in \mathbb{R}^2$, we condition on past observations (until $t_0$) at $\bm s_{0}$ only. In practice, we condition on the observations at $(\bm s_0, t_0)$ and $(\bm s_0, t_0-1)$ as it leads to the same forecast accuracy as when using more conditioning time points (not shown). For both models, we performed the conditional simulation using the \texttt{condrmaxstab} function from the \texttt{SpatialExtremes} \texttt{R} package \citep{PackageSpatialExtremes} that implements the method described in \citet{dombry2013conditional}.

To assess the quality of our predictions, we randomly chose 2000 (space-time) observations, for each of which we made 500 independent forecasts at the corresponding space-time points, transformed to the Gumbel scale by taking the logarithm. Each forecast is based on observations taken at least $u$ hours before the considered time point, and we repeated the experiment for $u\in\{1,\ldots,7\}$. We used the same random seed across all experiments to ensure consistency. Figure~\ref{Fig_ScatterPlots} shows that the quality of the forecasts worsens for both models as the time lag $u$ increases and that our model's forecasts are more in line with observed values, especially for large $u$, for which the DKS model's forecasts do not seem to relate to the observations. 

To objectify this visual impression, we employed two metrics: the root mean square error (RMSE) and the continuous ranked probability score \citep[CRPS;][]{Matheson1976}. The CRPS of a forecast for a space-time point $(\bm s, t)$ at lead time $u$ is
\begin{equation*}
    \mathrm{CRPS}(\bm{s},t,u) = \int_\R\left(\hat F_{\bm s, t, u}(x) - \mathbb{I}(x \geq \log Z_{\bm s,t})\right)^2\,\mathrm{d}x,
\end{equation*}
where $\log Z_{\bm s, t}$ is the observation on the Gumbel scale at $(\bm{s},t)$, and $\hat F_{\bm{s},t,u}$ is the empirical distribution function of the 500 independent forecasts, for $(\bm{s},t)$ at lead time $u$, transformed to the Gumbel scale.
We computed
the CRPS
using the \texttt{R} package \texttt{scoringRules} \citep{jordan2019}.

The plots in Figure~\ref{Fig_PredictScores} indicate that, as the time lag increases, our model outperforms the DKS one, although the performances of both models decrease. The ability of our model to appropriately account for temporal dependence (thanks to the explicit representation of the dynamics) significantly influences the forecasting skill for longer forecast horizons. The similarity of the plots in Figure~\ref{Fig_PredictScores} indicates that the difference in the two models' performances is consistent across various predictive scores, further validating the robustness of our results.

\begin{figure}[!h]
\centering 
\includegraphics[scale = 0.55]{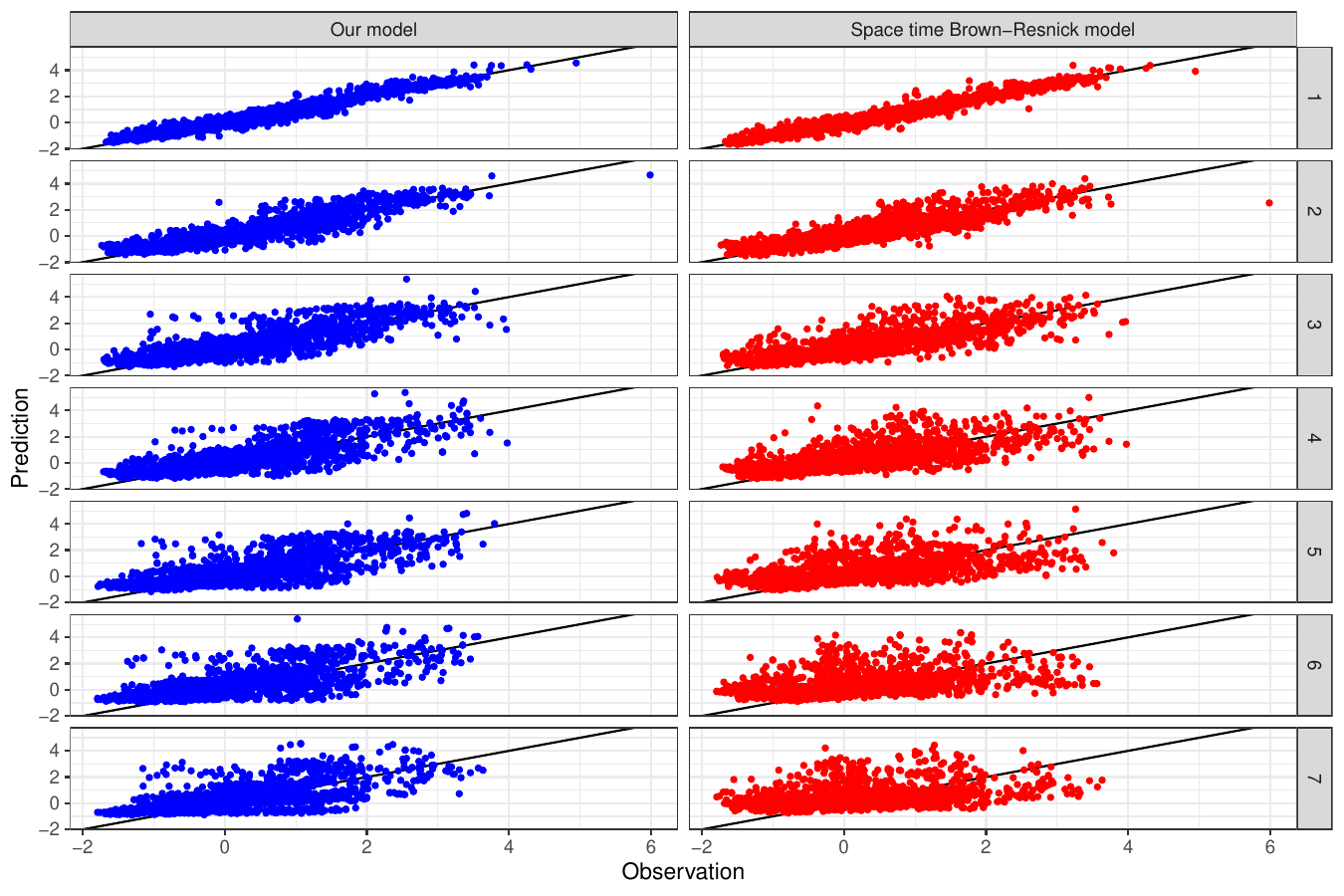}
\caption{Scatter plots of observations and forecasts (on the Gumbel scale) from our model (blue) and the DKS model (red) for time lags $u$ ranging from 1 to 7 hours (indicated at the end of each row).}
\label{Fig_ScatterPlots}
\end{figure}

\begin{figure}
    \centering 
    \includegraphics[width=0.49\textwidth]{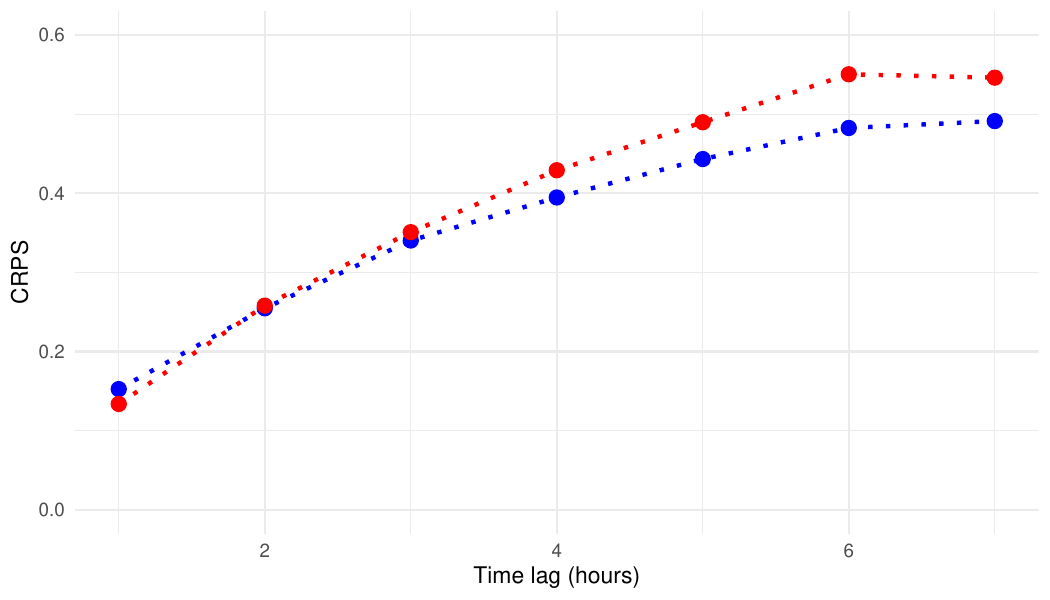}
    \includegraphics[width=0.49\textwidth]{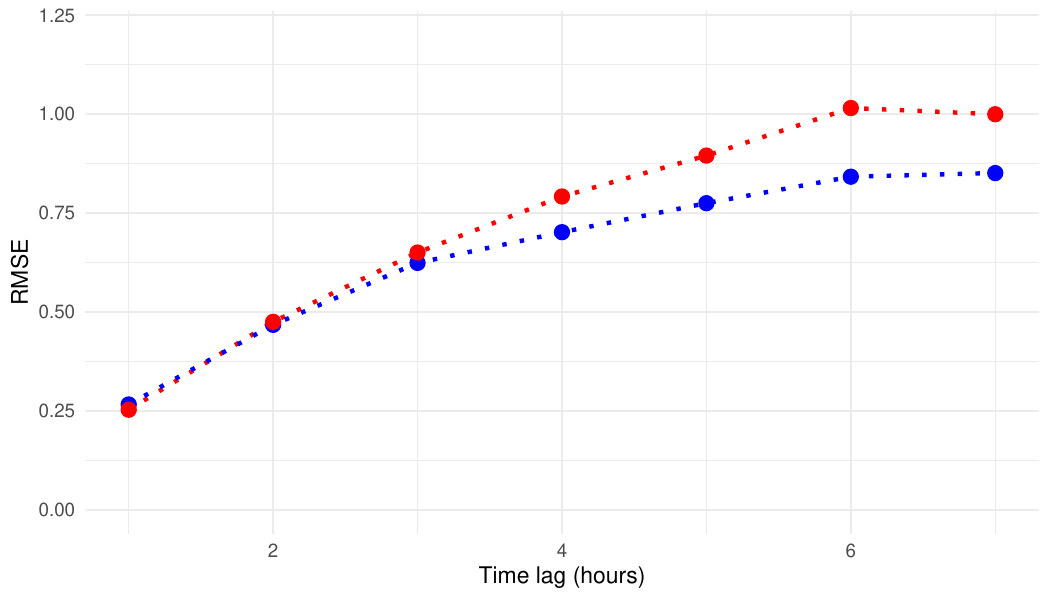}
    \caption{Mean CRPS (left) and RMSE of the mean of the ensemble of 500 independent forecasts (right) from our model (blue) and the DKS one (red), computed over the 2000 values for each lag $u=1, \ldots, 7$.}
    \label{Fig_PredictScores}
\end{figure}

\section{Discussion}
\label{Sec_Discussion}

Our focus is on the forecast of hourly maxima of 3-second wind gust speeds, as they are key indicators of potential associated damage. As explained, from a theoretical point of view, space-time max-stable models are natural for this task.

We focus on a specific space-time max-stable model which is Markovian in time (owing to its max-autoregressive structure) and includes an advection component, making it particularly suited for the forecasting (especially nowcasting) of atmospheric phenomena. We thoroughly analyze the theoretical properties of the model as well as those of the pairwise likelihood estimator (consistency and asymptotic normality), detail our forecasting strategy, and show the performance of our approach on wind gust reanalysis data for a period of a few days over Northwestern France in December 1999. Our methodology shows satisfactory results on this dataset, and demonstrates broad applicability beyond the current context, being suitable for forecasting wind gust speeds in other seasons (e.g., during thunderstorms in summer) and adaptable to other meteorological variables such as temperature, rainfall, or pollutant concentration. On top of tackling a prominent practical problem, we add to the limited literature about forecasting for max-stable fields. 

This work contributes to the field of statistical weather forecasting and provides a complementary approach to numerical weather prediction (NWP)- and artificial intelligence (AI)-based methodologies. Our model is parsimonious, enables easy quantification of the parameters' uncertainty and, owing to its explicit dynamics, is interpretable (causal representation) and allows straightforward ensemble forecasting. The price to pay for parsimony is a lack of flexibility compared to NWP or AI-based approaches. Especially, the current version of the model assumes spatially and temporally constant values of the decay and advection parameters $a$ and $\bm \tau$, implying that it does not guarantee reliable forecasts for large lead times (greater than a few hours or a day) and must be fitted to past data associated with very similar conditions. Two strategies can be employed to calibrate the model. If the synoptic (large-scale) conditions at the time of the forecast are comparable to those dominating in the previous hours or days (so that the flow direction is unaltered), the model can be fitted to the data collected on that period. An alternative approach involves fitting the model to concatenated historical data associated with similar weather regimes (i.e., quasi-stationary, persistent and recurrent large-scale flow patterns in the mid-latitudes; see \cite{grams2017balancing}, \cite{mockert2023meteorological}, and references therein) to the one prevailing at the time of the forecast. Moreover, the model presented in that paper concerns the spatially-standardized (to the Fréchet scale) data, but in practice the spatial margins also need to be modeled, possibly with non-stationarity included. 
For detailed information about the practical use of our model for weather forecasts, see Section~\ref{App_Operational}. Note that, in some settings \citep[see, e.g.,][]{weber1998}, geological features restrict the direction of advection to specific orientations with low variability, justifying the assumption of a constant $\bm \tau$, and thus enhancing the applicability of our model.

Future research directions to make our model more flexible include treating the decay and advection parameters $a$ and $\bm \tau$ as random or allowing them to depend on space, time, and atmospheric covariates (e.g., geopotential heights at various pressure levels, temperature and humidity at different elevations, radar and satellite data, lightning maps) to account for varying advection patterns over large spatial domains and through time. The inclusion of this information could be done through AI-based tools such as neural networks, and this would leverage the flexibility of AI while keeping the theoretically sound structure of our model. The space-time stationarity of the spatial dependence structure could also be relaxed in both space  and time, using, e.g., the approaches of \citet{Huser2016} and \citet{kohspacetimethunderstorms}, respectively, or even AI-based versions of the methods presented in those papers.
Some more flexibility could also be added by including a noise term to the recurrence equation in~\eqref{eqn:max_auto}, following common practices in econometrics. 
Finally, we could envisage the forecasts stemming from our model being post-processed by AI-based models. All these adjustments would enhance the model's ability to capture complex weather patterns and improve its forecasting skills across diverse atmospheric conditions and geographical regions, while retaining the interpretability, the explicit representation of the temporal dynamics and the facility to make ensemble forecasts.

{
    \spacingset{1}
    \bibliographystyle{apalike}
    \bibliography{References_Erwan}
}

\newpage
\appendix
\section*{SUPPLEMENTARY MATERIAL}

\renewcommand{\thesection}{\Alph{section}}

\section{Operational use of our model}\label{App_Operational}


In this section, we outline how our model may be used in practice to forecast weather phenomena in a context where historical measurements of the relevant quantity are available. The procedure can be decomposed into the following steps.

\begin{enumerate}
\item \textbf{Choose the calibration period (using historical periods with similar weather patterns)}. 

Determine the synoptic weather conditions at the time of the forecast by, e.g., specifying the weather regime (see, e.g., \cite{grams2017balancing}, \cite{mockert2023meteorological}, and references therein). A relevant example of weather regime for this study is the positive phase of the North Atlantic Oscillation (NAO+) or zonal regime, which is characterized by a pronounced positive pressure difference between the Azores High and the Icelandic Low, leading to strong westerly winds and possibly the formation of extratropical cyclones (winter storms). Other examples of weather regimes are the negative phase of the North Atlantic Oscillation (NAO-) and blocking regimes. Note that a quite fine classification is needed in our case, especially for local phenomena occurring in a summer. Once the current synoptic conditions have been specified, there are two options for the calibration:
\begin{enumerate}
\item If these conditions persist since a few days, fit the model to the associated dataset.
\item Otherwise, identify periods in the historical record with very similar synoptic conditions (e.g., the same weather regime) to the ones prevailing at the time of the forecast, and concatenate the associated data to create a ``historical dataset'' on which the model can be calibrated. 
However, care should be taken to mark the times at which the data were concatenated, as there is a break in the temporal dependence structure at those.
\end{enumerate}
    
\item \textbf{For each spatial coordinate, determine the parameters of the GEV distribution that best model the data}. As outlined in the first paragraph of Section~\ref{Sec_CaseStudy}, for each spatial point $\bm s$, use maximum likelihood estimation to fit the GEV distribution to the collection of observations in the historical dataset at $\bm s$. Let the resulting GEV distribution function at $\bm s$ be denoted $\hat F_{\bm s}$.
    
\item \textbf{Transform the historical data to have standard Fréchet margins using the estimated GEV distribution}. This is achieved by transforming the vector of observations at site $\bm s$ by applying the transformation
$$T_{\bm s}:x\mapsto -\frac{1}{\log \hat{F}_{\bm s}(x)}$$
    to each component of the vector, i.e., at each time point. This step is performed for each site $\bm s$ in the historical dataset, so that the data at all space-time points are transformed by the appropriate $T_{\bm s}$.
    
\item \textbf{Estimate the model parameters from the transformed historical data using the method outlined in Section~\ref{Sec_InferenceSimulations}}. We recall the one-step pairwise likelihood estimation strategy, in which~\eqref{eqn:pl-likelihoood} is maximized with respect to the parameter vector $\bm \psi$. A consequence of the concatenation of historical events in Step~1 is that temporal dependence is broken between the concatenated events. Thus, it is important to exclude pairs of space-time points from the likelihood function whenever $t$ and $t+u$ are not from the same event. Maximizing this censored pairwise likelihood function with respect to $\bm \psi$ yields estimates for the range parameter $\kappa$, the smoothness parameter $2H$, the advection parameter $\bm \tau$, and the decay constant $a$. We may denote the maximizer as $\hat{\bm \psi}$. Alternatively, the estimation can be performed in two steps as explained in Section~\ref{Sec_InferenceSimulations}.
    
\item \textbf{Transform the most recent map of weather data (initial conditions) using the transformations $T_{\bm s}$ for each site $\bm s$}. In the same way that the historical data were transformed in Step~3, transform all of the observations at the forecasting time $t_0$ using the appropriate $T_{\bm s}$. That is, the observation $X_{\bm s, t_0}$ at the space-time point $(\bm s, t_0)$ is transformed to $Z_{\bm s, t_0} = T_{\bm s}\big(X_{\bm s, t_0}\big)$ for all $\bm s$ in the spatial domain.
    
\item \textbf{Perform the forecasting strategy detailed in Section~\ref{Sec_ModelForecasting} using the transformed observations $Z_{\bm s, t_0}$ and our model parameterized by $\hat{\bm \psi}$}. This step leverages the Markovian property of our model and only uses the transformed observations $Z_{\bm s, t_0}$, for $\bm s$ in the spatial domain. Recall that in Section~\ref{Sec_CaseStudy_Forecasting}, it is explained that in order to forecast at a space-time point $(\bm s, t_0 + u)$, the transformed observation $Z_{\bm s - u\hat{\bm\tau}, t_0}$ is needed. If there are no data recorded at the site $\bm s - u\hat{\bm\tau}$, then synthetic data may be simulated conditionally on four nearby points (the details are shown in Figure~\ref{Fig_s_min_utau}). Also, as explained above, this methodology allows one to get an ensemble of forecasts, which is highly valuable in the context of weather forecasting.

\item \textbf{Transform the forecasts back from standard Fréchet margins using the GEV parameters}. Apply $T^{-1}_{\bm s}$ to the forecast at $(\bm s, t_0 + u)$ to transform it back to the scale of interest.
\end{enumerate}

The constraints in the choice of the calibration dataset described in Step~1 are important in the current version of our model due to the parameters $a$ and $\bm \tau$ being spatially and temporally constant and the spatial dependence structure being stationary in space and time. These could however be relaxed in the more flexible versions of the model mentioned in Section~\ref{Sec_Discussion}. Note also that the forecasts can be updated in real time, as the data corresponding to the new initial conditions (see Step~5) arrive.

\section{Supplementary technical results}

\subsection{Proof of \eqref{Eq_RecurrenceLagu}}\label{Sec:proof_of_closed_form}

Let $\bm{s} \in \mathbb{R}^2$, $t\in \N^+$ and $u \in \{1,\ldots,t\}$.
By recursively applying \eqref{eqn:max_auto} starting from $Z(\bm s - u\bm \tau, t- u)$, one obtains
$$Z(\bm s - (u-1)\bm \tau, t- (u - 1)) = \max\{a Z(\bm s - u\bm \tau, t- u), (1-a)W_{t-u+1}(\bm s - (u-1)\bm \tau)\},$$
\begin{align*}
    Z(\bm s - (u-2)\bm \tau, t- (u - 2)) = \max\{a^2 Z(\bm s - u\bm \tau, t- u),\ & a(1-a)W_{t-u+1}(\bm s - (u-1)\bm \tau),\\
    &(1-a)W_{t-u+2}(\bm s - (u-2)\bm \tau)\},
\end{align*}
$$\ldots$$
\begin{align*}
    Z(\bm s, t) = \max\{a^u Z(\bm s - u\bm \tau, t- u),\ & a^{u-1}(1-a)W_{t-u+1}(\bm s - (u-1)\bm \tau),\\
    &a^{u-2}(1-a)W_{t-u+2}(\bm s - (u-2)\bm \tau),\\
    &\ldots\\
    &(1-a)W_t(\bm s)\},
\end{align*}
which simplifies to
$$Z(\bm s, t) = \max\left\{a^u Z(\bm s - u\bm \tau, t- u), (1-a)\bigvee_{k=0}^{u-1}a^kW_{t-k}(\bm s - k\bm\tau)\right\}.$$
This proves~\eqref{Eq_RecurrenceLagu}, since
$$(1-a)\bigvee_{k=0}^{u-1}a^kW_{t-k}(\bm s - k\bm\tau) = (1-a^u)\widetilde W_{t-u+1}^t(\bm s),$$
by definition in~\eqref{eqn:W_tilde}.

By stationarity, independence and simple max-stability of the spatial fields $(W_t)_{t\in\N}$, one has
$$\widetilde W_{t_1}^{t_2} \stackrel{\mathrm d}= \left(\frac{1-a}{1-a^{t_2-t_1}}\sum_{k=0}^{t_2-t_1-1}a^k\right)W(\bm s) = W(\bm s),\qquad \bm s \in \R^2.$$

\subsection{Bivariate density functions and their derivatives}

\label{sec:density_function}

When $\bm{h}\neq u\bm{\tau }$ and $Z$ is distributed according to our model with $W$ being the spatial Brown--Resnick model and parameter vector $\bm \psi$, the bivariate density function of $%
\left( Z(\bm{s},t),Z(%
\bm{s}+\bm{h},t+u)\right)^{\prime} $ is
\begin{equation}\label{eqn:log_f_in_appendix}
f_{\bm{h},u}(z_{1},z_{2};{\bm{\psi }})=\exp \left( V\left(
z_{1},z_{2}\right) +\log (V_{1}\left( z_{1},z_{2}\right) V_{2}\left(
z_{1},z_{2}\right) -V_{12}\left( z_{1},z_{2}\right) )\right) ,\quad
z_{1},z_{2}>0,
\end{equation}
with 
\begin{eqnarray*}
V\left( z_{1},z_{2}\right) &=&V_{Z;\bm{h},u}\left( z_{1},z_{2}\right)
,\quad V_{1}\left( z_{1},z_{2}\right) =\frac{\partial V\left(
z_{1},z_{2}\right) }{\partial z_{1}}, \\
V_{2}\left( z_{1},z_{2}\right) &=&\frac{\partial V\left( z_{1},z_{2}\right) 
}{\partial z_{2}},\quad V_{12}=\frac{\partial ^{2}V\left( z_{1},z_{2}\right) 
}{\partial z_{1}\partial z_{2}}.
\end{eqnarray*}%
Let%
\begin{align*}
q_{1}& =\frac{\log (z_{2}/(a^{u}z_{1}))}{\sqrt{2\gamma(\bm{h}%
-u\bm{\tau })}}+\sqrt{\frac{\gamma (\bm{h}-u%
\bm{\tau })}{2}}, \\
q_{2}& =\frac{\log (a^{u}z_{1}/z_{2})}{\sqrt{2\gamma(\bm{h}-u%
\bm{\tau })}}+\sqrt{\frac{\gamma (\bm{h}-u\bm{%
\tau })}{2}},
\end{align*}
where $\gamma$ is the semivariogram.
Then
\begin{align*}
V_{1}& =-\frac{1}{z_{1}^{2}}\Phi (q_{1})+\frac{1}{z_{1}}\varphi (q_{1})\frac{%
\partial q_{1}}{\partial z_{1}}+\frac{a^{u}}{z_{2}}\varphi (q_{2})\frac{%
\partial q_{2}}{\partial z_{1}}, \\
V_{2}& =\frac{1}{z_{1}}\varphi (q_{1})\frac{\partial q_{1}}{\partial z_{2}}-%
\frac{a^{u}}{z_{2}^{2}}\Phi (q_{2})+\frac{a^{u}}{z_{2}}\varphi (q_{2})\frac{%
\partial q_{2}}{\partial z_{2}}-\frac{1}{z_{2}^{2}}+\frac{a^{u}}{z_{2}^{2}},
\\
V_{12}& =-\frac{1}{z_{1}^{2}}\varphi (q_{1})\frac{\partial q_{1}}{\partial
z_{2}}-\frac{q_{1}}{z_{1}}\varphi (q_{1})\frac{\partial q_{1}}{\partial z_{1}%
}\frac{\partial q_{1}}{\partial z_{2}}-\frac{a^{u}}{z_{2}^{2}}\varphi (q_{2})%
\frac{\partial q_{2}}{\partial z_{1}}-\frac{a^{u}q_{2}}{z_{2}}\varphi (q_{2})%
\frac{\partial q_{2}}{\partial z_{1}}\frac{\partial q_{2}}{\partial z_{2}},
\end{align*}%
where $\varphi (\cdot )$ denotes the standard normal probability
density function. The partial derivatives of $q_{1}$ and $q_{2}$ with
respect to $z_{1}$ and $z_{2}$ can be expressed as 
\begin{align*}
\frac{\partial q_{1}}{\partial z_{1}}& =\frac{-1}{z_{1}\sqrt{2\gamma
(\bm{h}-u\bm{\tau })}},\quad \frac{\partial q_{1}}{%
\partial z_{2}}=\frac{1}{z_{2}\sqrt{2\gamma (\bm{h}-u%
\bm{\tau })}}, \\
\frac{\partial q_{2}}{\partial z_{1}}& =\frac{1}{z_{1}\sqrt{2\gamma(\bm{h}-u\bm{\tau })}},\quad \frac{\partial q_{2}}{\partial
z_{2}}=\frac{-1}{z_{2}\sqrt{2\gamma(\bm{h}-u\bm{\tau }%
)}},
\end{align*}%
and those with respect to $\gamma$ and $a$
are
\begin{equation*}
\frac{\partial q_{1}}{\partial \gamma }=\frac{q_{2}}{2\gamma },\quad \frac{%
\partial q_{1}}{\partial a}=-\frac{u}{2a\gamma },\quad \frac{\partial q_{2}}{%
\partial \gamma }=\frac{q_{1}}{2\gamma },\quad \frac{\partial q_{2}}{%
\partial a}=\frac{u}{2a\gamma }.
\end{equation*}%
We include the expressions of these partial derivatives for future reference in the proof of Lemma~\ref{lem:score_L3} below. It is important to note that the first and second partials of $f_{\bm{h},u}(z_{1},z_{2};{\bm{\psi }})$ with respect to $a$ are bounded
above on $\Psi _{\varepsilon }$ in~\eqref{eqn:psi_epsilon}. The first and second partials with respect to the remaining parameters in $\bm\psi$ act through $\gamma $,
and they too are bounded above on $\Psi_{\varepsilon }$.

\subsection{Constraints on the design mask $\mathcal H_r$}\label{sec:design_mask_app}

In the purely spatial setting, it is common practice to exclude the negatives of the vectors in the design mask $\mathcal{H}_r$, \textit{i.e.}, if $\bm{v} \in \mathcal H_r\setminus\{\bm 0\}$, then $-\bm v \notin \mathcal H_r$. This ensures that each pair of points in the dataset is considered in the pairwise log-likelihood estimator at most once. In the space-time setting, this restriction on $\mathcal H_r$ can be relaxed if one considers only positive temporal lags. Indeed, if one constructs pairs of observations by pairing a first space-time coordinate with another space-time coordinate at a later time, it is impossible for any pair of observations to be counted more than once. These considerations are especially pertinent if the model is not invariant to rotations in space, which is indeed the case for our model when $\bm \tau \neq \bm 0$.

\section{Asymptotic properties of the pairwise likelihood estimator}\label{Sec_Appendix_TheoreticalProperties}
\label{sec:pwle}

In this section, we prove that the pairwise likelihood estimator of the parameter vector ${%
\bm{\psi }}^{\star }$ is almost surely consistent and asymptotically
normal. The parameter space $\Psi _{\varepsilon }$ is compact and is assumed
to contain the true parameter vector ${\bm{\psi }}^{\star }$, i.e., the following condition
holds.

\begin{assu}
\label{ass:case1}
The parameter vector ${\bm{\psi }}^{\star }=(\bm{\theta }^{\star },
\bm{\tau }^{\star },a^{\star })^{\prime}$ lies in $\Psi _{\varepsilon }$ for
some $0<\varepsilon <\min \{ 1/2,\mu /p \}$.
\end{assu}

The elements of $\Psi _{\varepsilon }$ are identifiable in the sense that 
\begin{equation*}
{\bm{\psi }}=\Tilde{{\bm{\psi }}}\iff f_{\bm{h}%
,u}(z_{1},z_{2};{\bm{\psi }})=f_{\bm{h},u}(z_{1},z_{2};\Tilde{{
\bm{\psi }}}),
\end{equation*}
for all $\bm{h}\in \mathcal{H}_{r}$, $u\in \{1,\ldots,p\}$, and $%
z_{1},z_{2}>0$. Section~\ref{sec:density_function} provides the expression
of $f_{\bm{h},u}(z_{1},z_{2};{\bm{\psi }})$ and some
explanations for deriving its first and second derivatives with respect to ${%
\bm{\psi }}$ using the chain rule. It is in particular easily
deduced from this subsection that the pairwise likelihood function $\mathrm{%
PL}^{(m,T)}$ defined in \eqref{eqn:pl-likelihoood} and its first and second derivatives with respect to ${%
\bm{\psi }}$ are continuous in ${\bm{\psi }}$ on $\Psi
_{\varepsilon }$. Appendices~\ref{sec:pwle_consistency} and~\ref%
{sec:pwle_assnorm} provide proofs of the asymptotic consistency and
normality.

\subsection{Consistency of the pairwise likelihood estimator}
\label{sec:pwle_consistency}

\begin{theo}
\label{thm:consistency} Suppose Assumption \ref{ass:case1} holds. Then the
pairwise likelihood estimator in~\eqref{eqn:pl-estimator} satisfies
\begin{equation}
\hat{{\bm{\psi }}}=\arg \max_{{\bm{\psi }}\in \Psi
_{\varepsilon }}\mathrm{PL}^{(m,T)}({\bm{\psi }})\overset{\mathrm{%
a.s.}}{\longrightarrow }{\bm{\psi }}^{\star },  \label{eqn:psi_hat}
\end{equation}%
as $m,T\rightarrow \infty $.
\end{theo}

Before we begin the proof of Theorem \ref{thm:consistency}, we need the two
following lemmas.

\begin{lemm}
\label{lem:log_f_integrable} Under Assumption \ref{ass:case1}, the random
variable appearing in~\eqref{eqn:pl-likelihoood},
\begin{equation*}
\log f_{\bm{h},u}\big(Z(\bm{0},0),Z(\bm{h},u);{\bm{\psi }%
}\big),
\end{equation*}%
is uniformly integrable on $\Psi _{\varepsilon }$, for all $\bm{h}\in 
\mathcal{H}_{r}$, and $u\in \{1,\ldots,p\}$.
\end{lemm}

\begin{proof}
By \eqref{eqn:log_f_in_appendix}, the absolute value of the likelihood function is bounded as follows:
\begin{equation*}
|\log f_{\bm{h},u}(z_{1},z_{2};{\bm{\psi }})|\leq |V|+|\log
(V_{1}V_{2}-V_{12})|.
\end{equation*}
The two terms on the right-hand
side are considered separately. To bound $|V|$, recognize that $\Phi (\cdot
)\leq 1$. Thus, we have for all ${\bm{\psi }}\in \Psi _{\varepsilon }$, 
\begin{equation*}
V\big(Z(\bm{0},0),Z(\bm{h},u)\big)\leq \frac{1}{Z(\bm{0},0)}+%
\frac{1}{Z(\bm{h},u)},
\end{equation*}%
and 
\begin{equation*}
\E\bigg[\frac{1}{Z(\bm{0},0)}\bigg]+\E\bigg[\frac{1}{Z(\bm{h},u)}%
\bigg]<\infty ,
\end{equation*}%
since the space-time field $1/Z$ has exponential---thus integrable---margins.

What remains to be shown is that $\log (V_{1}V_{2}-V_{12})$ evaluated at $%
z_{1}=Z(\bm{0},0)$ and $z_{2}=Z(\bm{h},u)$ is uniformly integrable
over $\Psi _{\varepsilon }$ whenever $\bm{h}\in \mathcal{H}_{r}$ and $%
u\in \{1,\ldots,p\}$. For any choice of $z_{1}$ and $z_{2}$, $|\log
(V_{1}V_{2}-V_{12})|$ can be bounded above by a constant, since $a$ and $%
\gamma(\bm{h}-u\bm{\tau })$ can be
bounded away from $0$ on $\Psi_\varepsilon$. Therefore, for any $%
k_{1},k_{2}\in \R^{+}$ such that $k_{1}\leq k_{2}$, the quantity 
\begin{equation*}
K_{k_{1},k_{2}}=\sup \{|\log (V_{1}V_{2}-V_{12})|:z_{1},z_{2}\in \lbrack
k_{1},k_{2}],\bm{\psi }\in \Psi _{\varepsilon }\}
\end{equation*}%
exists and is finite.

Now, we consider the asymptotic behavior of $|\log (V_{1}V_{2}-V_{12})|$ in
the four cases $z_{1}\rightarrow \infty $, $z_{1}\rightarrow 0$, $%
z_{2}\rightarrow \infty $, and $z_{2}\rightarrow 0$. Referring to the
expressions in Section~\ref{sec:density_function}, it can be shown that
independently of ${\bm{\psi }}$, there exists a sufficiently large $%
k_{2}\in \R^{+}$, a sufficiently small $k_{1}\in \R^{+}$ such that $%
k_{1}\leq k_{2}$, and a sufficiently large $k_{3}\in \R^{+}$, such that 
\begin{equation}
|\log (V_{1}V_{2}-V_{12})|\leq \Big(\log z_{1}+\frac{1}{z_{1}}\Big)^{k_{3}}%
\Big(\log z_{2}+\frac{1}{z}\Big)^{k_{3}}  \label{eqn:logx_bound}
\end{equation}%
whenever $(z_{1},z_{2})\notin \lbrack k_{1},k_{2}]^{2}$.

Define $k_{1}$, $k_{2}$, and $k_{3}$ such that \eqref{eqn:logx_bound} holds.
In this case, 
\begin{equation*}
|\log (V_{1}V_{2}-V_{12})|\leq \Big[\sqrt{K_{k_{1},k_{2}}}+\Big(\log z_{1}+%
\frac{1}{z_{1}}\Big)^{k_{3}}\Big]\Big[\sqrt{K_{k_{1},k_{2}}}+\Big(\log z_{2}+%
\frac{1}{z_{2}}\Big)^{k_{3}}\Big],
\end{equation*}%
for all $z_{1},z_{2}>0$.

Let $Z_{1}=Z(\bm{0},0)$ and $Z_{2}=Z(\bm{h},u)$. By H\"{o}lder's
inequality, 
\begin{align*}
\E\bigg[\Big[\sqrt{K_{k_{1},k_{2}}}& +\Big(\log Z_{1}+\frac{1}{Z_{1}}\Big)%
^{k_{3}}\Big]\Big[\sqrt{K_{k_{1},k_{2}}}+\Big(\log Z_{2}+\frac{1}{Z_{2}}\Big)%
^{k_{3}}\Big]\bigg] \\
\leq & \ \E\bigg[\Big[\sqrt{K_{k_{1},k_{2}}}+\Big(\log Z_{1}+\frac{1}{Z_{1}}%
\Big)^{k_{3}}\Big]^{2}\bigg]^{1/2}\E\bigg[\Big[\sqrt{K_{k_{1},k_{2}}}+\Big(%
\log Z_{2}+\frac{1}{Z_{2}}\Big)^{k_{3}}\Big]^{2}\bigg]^{1/2} \\
<& \ \infty ,
\end{align*}%
since the margins of $\log Z$ and $1/Z$ are the Gumbel and exponential
distributions respectively, which have finite moments.

We have shown that for all $(\bm{h},u)\in \mathcal{H}_{r}\times
\{1,\ldots,p\}$, the quantity $\big|\log f_{\bm{h},u}\big(Z(\bm{0},0),Z(%
\bm{h},u);{\bm{\psi }}\big)\big|$ can be bounded above by the
sum of two integrable random variables that do not depend on ${\bm{%
\psi }}$, implying uniform integrability.
\end{proof}

To state the second lemma, we first need to introduce the following definition.
\begin{defi}[Space-time mixing]\label{def:mixing_in_space_and_time}
An $\R$-valued space-time field $Z$ is said to be {space-time mixing} if, for any $\bm{s}\in\R^2$, any $t, z_1, z_2 \in \R$, and any sequences $(\bm h_n)_{n\geq 1}\in \R^2$ and $(u_n)_{n\geq 1}\in \R$ satisfying $\max\big\{||\bm{h}_n||,u_n\big\}\rightarrow\infty$ as $n\to \infty$, it holds that
\begin{equation*}
\P\big[Z(\bm{s},t)\leq z_1, Z(\bm{s}+\bm{h}_n,t+u_n)\leq z_2\big] - \P\big[Z(\bm{s},t)\leq z_1\big]\P\big[Z(\bm{s}+\bm{h}_n,t+u_n)\leq z_2\big]\xrightarrow[n\to\infty]{}0.
\end{equation*}
\end{defi}

\begin{lemm}\label{lem:X_mixing}
Let $Z$ be distributed according to our model in~\eqref{eqn:max_auto} with $W$ being the spatial Brown--Resnick model. The bivariate extremal dependence coefficient $\Theta_Z(\bm h, u)$ satisfies both
\begin{enumerate}
    \item[(a)] $\inf_{\bm h \in \R^2}\Theta_Z(\bm h, u) = 2 - a^u$,
    \item[(b)] $\inf_{u\geq 0}\Theta_Z(\bm h + u\bm\tau, u) = \Theta_Z(\bm h, 0)$.
\end{enumerate}
Moreover, the field $Z$ given by \eqref{eqn:max_auto}
is space-time mixing as in Definition~\ref{def:mixing_in_space_and_time} for any field $W$ that is mixing.
\end{lemm}

\begin{proof}
We begin by proving Item~(a).
For fixed $u\geq 0$, let $\bm{h}\in \R^{2}\setminus \{u\bm{\tau }\}$ and define $x=\sqrt{\gamma(\bm{h}-u\bm{\tau }%
)/2}>0$. Then
\begin{align*}
\Theta_Z (\bm{h},u)=& \ V_{Z;\bm{h},u}(1,1)= \ \Phi \bigg(x-\frac{u\log a}{2x}\bigg) +a^{u}\Phi \bigg(x+\frac{u\log a}{2x}\bigg) +1-a^{u}.
\end{align*}
We exploit the identities
\begin{equation*}
1-\Phi \bigg(x-\frac{u\log a}{2x}\bigg)=\int_{0}^{\infty }\varphi \bigg(t+x-%
\frac{u\log a}{2x}\bigg)\,\mathrm{d} t
\end{equation*}%
and 
\begin{equation*}
a^{u}\Phi \bigg(x+\frac{u\log a}{2x}\bigg)=\int_{0}^{\infty }a^{u}\varphi %
\bigg(-t+x+\frac{u\log a}{2x}\bigg)\,\mathrm{d}t
\end{equation*}
to write
\begin{align*}
\Theta_Z (\bm{h},u)-\big(2-a^{u}\big)=& \ a^{u}\Phi \bigg(x+\frac{u\log a}{2x}%
\bigg)+\Phi \bigg(x-\frac{u\log a}{2x}\bigg)-1 \\
=& \ \int_{0}^{\infty }\bigg[a^{u}\varphi \bigg(-t+x+\frac{u\log a}{2x}\bigg)%
-\varphi \bigg(t+x-\frac{u\log a}{2x}\bigg)\bigg]\,\mathrm{d}t \\
>& \ 0.
\end{align*}%
The inequality holds since the integrand is positive for all $t>0$ due to
the identity 
\begin{equation}\label{eqn:identity_phi}
a^{u}\varphi \bigg(-t+x+\frac{u\log a}{2x}\bigg)= e^{2tx}\varphi \bigg(%
t+x-\frac{u\log a}{2x}\bigg),\qquad x\neq 0, 
\end{equation}
which can be shown by taking logarithms and expanding the squared trinomials.
Therefore, for any $\bm{h}\in \R^{2}\setminus\{u\bm \tau\}$, 
\begin{equation}\label{eqn:bound_2-au}
\Theta_Z (\bm{h},u)\geq 2-a^{u},  
\end{equation}%
and as seen previously in Section~\ref{Subsec_Model}, $\Theta_Z (u\bm \tau,u) = 2-a^{u}$. This proves Item~(a).

Item~(b) holds trivially if $\bm{h}=
\bm{0}$, in which case $\Theta_Z (\bm{0},0)=1$. Fix $\bm{h}\neq 
\bm{0}$ and redefine $x=\sqrt{\gamma(\bm{%
h})/2}$. Then for $u \geq 0$, we have 
\begin{align*}
\frac{\partial }{\partial u}\Theta_Z (\bm{h}+u\bm{\tau},u)=& -%
\frac{\log a}{2x}\varphi \bigg(x-\frac{u\log a}{2x}\bigg)+a^{u}\log a\Phi %
\bigg(x+\frac{u\log a}{2x}\bigg) \\
& +\frac{a^{u}\log a}{2x}\varphi \bigg(x+\frac{u\log a}{2x}\bigg)-a^{u}\log a
\\
\geq & \ \frac{\log a}{2x}\bigg\{a^{u}\varphi \bigg(x+\frac{u\log a}{2x}%
\bigg)-\varphi \bigg(x-\frac{u\log a}{2x}\bigg)\bigg\} \\
=& \ 0.
\end{align*}%
We remind the reader that the inequality holds since $%
\log a<0$, and $\Phi (\cdot )\leq 1$. The last equality follows from %
\eqref{eqn:identity_phi}. Finally, by the fundamental theorem of calculus, 
\begin{equation*}
\Theta_Z (\bm{h} + u\bm \tau,u)-\Theta_Z (\bm{h},0)=\int_{0}^{u}%
\frac{\partial }{\partial \tilde{u}}\Theta_Z \big(\bm{h}+\tilde{u}
\bm{\tau },\tilde{u}\big)\,\mathrm{d}\tilde{u}\geq 0,
\end{equation*}%
which proves (b).

Now, to see that $Z$ is space-time mixing, let $({\bm h}_n)_{n\geq 1} \in \R^2$ and $(u_n)_{n\geq 1}$ such that $M_n=\max\{ ||\bm{h}_n||,u_n\} \xrightarrow[n\to\infty]{}\infty$. It suffices to show that $\Theta_Z(\bm h_n, u_n)\xrightarrow[n\to\infty]{}2$.
Let $n\in\N$, and suppose that $u_n < M_n/(||\bm \tau|| + 1)$. Then,
\begin{align*}
    ||\bm h_n - u_n \bm\tau|| \geq ||\bm h_n|| - u_n ||\bm\tau|| > M_n - M_n\frac{||\bm\tau||}{||\bm \tau|| + 1} = \frac{M_n}{||\bm \tau|| + 1}.
\end{align*}
Thus, either $u_n \geq M_n/(||\bm \tau|| + 1)$ or $||\bm h_n - u_n \bm\tau|| \geq M_n/(||\bm \tau|| + 1)$. Hence,
\begin{equation}\label{eqn:same_rate_max}
\max\{u_n, ||\bm{h}_n - u_n \bm\tau||\} \geq \frac{M_n}{||\bm \tau|| + 1} \xrightarrow[n\to\infty]{}\infty.
\end{equation}
Now, by Items~(a) and (b),
$$\Theta_Z(\bm h_n, u_n) \geq \max\big\{2 - a^{u_n}, \Theta_Z(\bm{h}_n - u_n \bm\tau, 0)\big\}\xrightarrow[n\to\infty]{} 2,$$
since the innovation spatial field $W$ is mixing.
\end{proof}

\begin{proof}[Proof of Theorem \protect\ref{thm:consistency}]
We follow the method of proof demonstrated in \cite{davis2013b}. The pairwise likelihood
function defined in \eqref{eqn:pl-likelihoood} can be expressed as 
\begin{equation*}
\mathrm{PL}^{(m,T)}({\bm{\psi }})=\sum_{\bm{s}\in \mathcal{S}%
_{m}}\sum_{t=1}^{T}g_{r,p}(\bm{s},t;{\bm{\psi }})-\mathcal{R}%
^{(m,T)}({\bm{\psi }}),
\end{equation*}%
where 
\begin{equation}\label{eqn:grp_def}
g_{r,p}(\bm{s},t;{\bm{\psi }})=\sum_{\bm{h}\in \mathcal{H}%
_{r}}\sum_{u=1}^{p}\log f_{\bm{h},u}\big(Z(\bm{s},t),Z(\bm{s}+%
\bm{h},t+u);{\bm{\psi }}\big)
\end{equation}%
and 
\begin{align}
\mathcal{R}^{(m,T)}({\bm{\psi }})=& \ \sum_{\bm{s}\in
S_{m}}\sum_{t=1}^{T}\bigg(\sum_{\bm{h}\in \mathcal{H}_{r}}\sum 
_{\substack{ u=1  \\ t+u>T}}^{p}+\sum_{\substack{ \bm{h}\in \mathcal{H}%
_{r}  \\ \bm{s}+\bm{h}\notin \mathcal{S}_{m}}}\sum_{u=1}^{p}-\sum 
_{\substack{ \bm{h}\in \mathcal{H}_{r}  \\ \bm{s}+\bm{h}\notin 
\mathcal{S}_{m}}}\sum_{\substack{ u=1  \\ t+u>T}}^{p}\bigg)  \label{eqn:R^mt}
\\
& \ \log f_{\bm{h},u}\big(Z(\bm{s},t),Z(\bm{s}+\bm{h},t+u);{%
\bm{\psi }}\big).  \notag
\end{align}%
We continue by showing that in the limit as $m,T\rightarrow \infty $, 
\begin{equation}
\frac{1}{m^{2}T}\mathrm{PL}^{(m,T)}({\bm{\psi }})\overset{\mathrm{%
a.s.}}{\longrightarrow }\E[g_{ r, p}(\bm{s}_1, 1; {\bm\psi})],
\label{eqn:normalized_PL_converge}
\end{equation}%
uniformly on $\Psi _{\varepsilon }$. Lemmas \ref{lem:log_f_integrable} and %
\ref{lem:X_mixing} ensure that we can apply the uniform strong law of large numbers given by Theorem 2.7 in \cite%
{straumann2006} to the sequence $\big(g_{r,p}(\bm{s},t;{\bm{\psi 
}})\big)_{\bm{s}\in S_{m},t\in \{1,\ldots,T\}}$, which implies that as $m,T\rightarrow \infty $, 
\begin{equation*}
\frac{1}{m^{2}T}\sum_{\bm{s}\in \mathcal{S}_{m}}\sum_{t=1}^{T}g_{r,p}(%
\bm{s},t;{\bm{\psi }})\overset{\mathrm{a.s.}}{\longrightarrow }%
\E[g_{ r, p}(\bm{s}_1, 1; {\bm\psi})],
\end{equation*}%
uniformly on $\Psi _{\varepsilon }$. Equation %
\eqref{eqn:normalized_PL_converge} is then implied if 
\begin{equation}
\frac{1}{m^{2}T}\mathcal{R}^{(m,T)}({\bm{\psi }})\overset{\mathrm{%
a.s.}}{\longrightarrow }0,\qquad m,T\rightarrow \infty ,
\label{eqn:normalized_R_converge}
\end{equation}%
uniformly on $\Psi _{\varepsilon }$. Again, Theorem 2.7 in \cite%
{straumann2006} implies that there exists a constant $L<\infty $ such that 
\begin{equation*}
\frac{1}{m^{2}+mT}\mathcal{R}^{(m,T)}({\bm{\psi }})\overset{%
\mathrm{a.s.}}{\longrightarrow }L,\qquad m,T\rightarrow \infty ,
\end{equation*}%
uniformly on $\Psi _{\varepsilon }$, since the right-hand side of %
\eqref{eqn:R^mt} has on the order of $m^2+mT$ terms. Therefore, %
\eqref{eqn:normalized_R_converge} and \eqref{eqn:normalized_PL_converge}
hold. The convergence is uniform on $\Psi _{\varepsilon }$, so $\hat{{%
\bm{\psi }}}$ as defined in \eqref{eqn:pl-estimator} converges almost
surely to the ${\bm{\psi }}\in \Psi _{\varepsilon }$ that maximizes $%
\E[g_{ r, p}(\bm{s}_1, 1;
{\bm\psi})]$. 

Finally, an application of Jensen's inequality shows
that the true parameter vector ${\bm{\psi }}^{\star }$ for the random field $Z$
is the unique maximizer of $%
\E[g_{ r, p}(\bm{s}_1, 1;
{\bm\psi})]$. Indeed, for any ${\bm{\psi }}\in \Psi _{\varepsilon }$, 
\begin{align*}
\E[g_{ r, p}(\bm{s}_1, 1; {\bm\psi})]-\E[g_{ r, p}(\bm{s}_1,
1; {\bm\psi}^\star)]& =\sum_{\bm{h}\in \mathcal{H}_{r}}\sum_{u=1}^{p}%
\E\bigg[\log \bigg(\frac{f_{\bm{h},u}(Z_{1},Z_{2};{\bm{\psi }})}{%
f_{\bm{h},u}(Z_{1},Z_{2};{\bm{\psi }}^{\star })}\bigg)\bigg] \\
& \leq \sum_{\bm{h}\in \mathcal{H}_{r}}\sum_{u=1}^{p}\log \E\bigg[\frac{%
f_{\bm{h},u}(Z_{1},Z_{2};{\bm{\psi }})}{f_{\bm{h}%
,u}(Z_{1},Z_{2};{\bm{\psi }}^{\star })}\bigg] \\
& =\sum_{\bm{h}\in \mathcal{H}_{r}}\sum_{u=1}^{p}\log (1) \\
& =0,
\end{align*}%
where $Z_{1}=Z(\bm{s}_{1},1)$, and $Z_{2}=Z(\bm{s}_{1}+\bm{h}%
,1+u)$ are used for shorthand. Equality in Jensen's inequality holds if and
only if $f_{\bm{h},u}(Z_{1},Z_{2};{\bm{\psi }})=f_{\bm{h}%
,u}(Z_{1},Z_{2};{\bm{\psi }}^{\star })$ almost surely, which is
equivalent to ${\bm{\psi }}={\bm{\psi }}^{\star }$ by
identifiability. Therefore, ${\bm{\psi }}^{\star }$ is the unique
maximizer of $\E[g_{ r,
p}(\bm{s}_1, 1; {\bm\psi})]$, and thus the unique limiting value
of $\hat{{\bm{\psi }}}$ almost surely.
\end{proof}


\subsection{Asymptotic normality of the pairwise likelihood estimator}

\label{sec:pwle_assnorm}

We now study the asymptotic distribution of $\hat{{\bm{\psi}}}$ as $%
m,T\rightarrow\infty$. To show that a central limit theorem applies in our
setting, it is important to understand the rate at which dependence is lost
between two space-time points as they are separated in either space or time.
This information is encoded in the \textit{$\alpha$-mixing coefficients},
which are defined for the space-time field in \cite{davis2013b} as follows.

Define the distances 
\begin{equation}
d_{\infty}\big((\bm{s}_{1},t_{1}),(\bm{s}_{2},t_{2})\big)=\max \big\{||%
\bm{s}_{2}-\bm{s}_{1}||,|t_{2}-t_{1}|\big\},\qquad \bm{s}%
_{1},\bm{s}_{2}\in \mu\Z^{2},\ t_{1},t_{2}\in \N  \label{eqn:d_prime}
\end{equation}%
and 
\begin{equation*}
d(\Lambda _{1},\Lambda _{2})=\inf \big(d_{\infty}(\rho _{1},\rho _{2}):\rho
_{1}\in \Lambda _{1},\rho _{2}\in \Lambda _{2}\big),\qquad \Lambda
_{1},\Lambda _{2}\subset \mu\Z^{2}\times \N.
\end{equation*}%
Further, let $\mathcal{F}_{\Lambda_i }$ denote the $\sigma $-algebra generated
by $\{Z(\bm{s},t):(\bm{s},t)\in \Lambda_i \}$, $i=1,2$. Then for $n\in \N$ and $%
k,l\in \N\cup \{\infty \}$, the $\alpha$-mixing coefficients for $Z$ are defined as 
\begin{equation*}
\alpha _{k,l}(n)=\sup \{|\P (A_{1}\cap A_{2})-\P (A_{1})\P (A_{2})|:A_{i}\in 
\mathcal{F}_{\Lambda _{i}},|\Lambda _{1}|\leq k,|\Lambda _{2}|\leq
l,d(\Lambda _{1},\Lambda _{2})\geq n\}.
\end{equation*}

For a measurable function $h$, if $h\big(Z(\bm{s}_{1},1)\big)$ obeys
some specific moment conditions and the $\alpha $-mixing coefficients
decay sufficiently fast with $n$, then a central limit theorem can be
applied to samples of $h\big(Z(\bm{s},t)\big)$ at regularly spaced
intervals in space and time. Inspired by \cite{davis2013b}, we show that a
central limit theorem applies to the random field 
\begin{equation}
\{\nabla _{\bm{\psi }}g_{r,p}(\bm{s},t;{\bm{\psi }}%
^{\star })\}_{\bm{s}\in \mu\Z^{2},t\in \N},  \label{eqn:score}
\end{equation}%
which we use to show the asymptotic normality of $\hat{\bm{\psi 
}}$.

\begin{theo}
\label{thm:normality} Suppose Assumption \ref{ass:case1} holds. Then the pairwise likelihood
estimator $\hat{\bm{\psi }}$ defined in \eqref{eqn:pl-estimator} is
asymptotically normal in the sense that 
\begin{equation*}
(m^{2}T)^{1/2}(\hat{\bm{\psi }}-{\bm{\psi }}^{\star })%
\xrightarrow{\mathrm{d}}\mathcal{N}\big(0,F^{-1}\Sigma (F^{-1})^{\prime }%
\big),\qquad m,T\rightarrow \infty ,
\end{equation*}%
where 
\begin{equation*}
F=\E[-\nabla_{\bm\psi}^2 g_{ r, p}(\bm{s}_1, 1;
{\bm\psi}^\star)]
\end{equation*}%
and 
\begin{equation}
\Sigma =\sum_{\bm{s}\in \mu\Z^{2}}\sum_{t\in \N}\mathrm{Cov}[\nabla _{%
\bm{\psi }}g_{r,p}(\bm{s}_{1},1;{\bm{\psi }}^{\star
}),\nabla _{\bm{\psi }}g_{r,p}(\bm{s},t;{\bm{\psi }}%
^{\star })].  \label{eqn:big_sigma}
\end{equation}
\end{theo}

Even though $\nabla _{\bm{\psi }}g_{r,p}(\bm{s},t;{\bm{\psi 
}})$ is a function of $Z$ observed at multiple space-time points, the $
\alpha $-mixing coefficients for the field defined in \eqref{eqn:score}
decay at the same rate as the $\alpha $-mixing coefficients for $Z$ with
suitably rescaled values of $k$ and $l$, since the $\sigma$-algebra generated by $\{\nabla _{\bm{\psi }%
}g_{r,p}(\bm{s},t;{\bm{\psi }})\}$ is contained in $\mathcal{F}_{\Lambda
} $ for some $\Lambda \subseteq \mu\Z^{2}\times \N$. Therefore, the asymptotic
behavior of the $\alpha $-mixing coefficients for $\nabla _{\bm{%
\psi }}g_{r,p}(\bm{s},t;{\bm{\psi }})$ can be completely
understood from the following lemma.

\begin{lemm}
\label{lem:alpha_mixing} Under the conditions of Theorem \ref{thm:normality}%
, there exists a constant $\varepsilon >0$ such that the $\alpha $-mixing coefficients
for $Z$ satisfy 
\begin{equation}\label{eqn:statement_of_lemma_3}
\liminf_{n\to\infty} \frac{-\log \alpha _{k,l}(n)}{n^{\min\{2H^\star, 1\}}} > \varepsilon,
\end{equation}%
for all $k\in \N$ and $l\in \N\cup \{\infty \}$, where $H^\star$ is the Hurst index of the semivariogram (see Section~\ref{Sec_CaseStudy_Performance}).
\end{lemm}

\begin{proof}
We use Corollary 2.2 in \cite{dombry2012} to bound the $\alpha $-mixing
coefficients for $Z$ as follows:
\begin{align}
\alpha _{k,l}(n)& \leq 2kl\sup \big\{2-\Theta_Z (\bm{h},u):\max \{||\bm{h}%
||,u\}\geq n\big\}, \label{eqn:first_dombry_eyi}\\
\alpha _{k,\infty }(n)& \leq 2k^2\sum_{m=n}^{\infty }N(m)\sup
\big\{2-\Theta_Z (\bm{h},u):\max \{||\bm{h}||,u\}\geq m\big\},\label{eqn:second_dombry_eyi}
\end{align}%
where $N(m)$ is the number of points in $\mu\Z^2\times \N$ whose distance to the origin is between $m$ and $m+1$, which is of the order $\mathcal{O}(m^2)$. The inequality in~\eqref{eqn:second_dombry_eyi} holds because for any $\Lambda \subseteq \mu\Z%
^{2}\times \N$ such that $|\Lambda |=k$, there are at most $kN(m)$
points in $\mu\Z^{2}\times \N$ whose distance to the closest point in $\Lambda $ is between $m$ and $m+1$.

Using the same arguments as those leading to~\eqref{eqn:same_rate_max}, we can show that, for any $n \in \N$,
$$\max\{||\bm h||, u\} \geq n \Longrightarrow \max\{||\bm h - u\bm\tau^\star||, u\} \geq \frac{n}{1 + ||\bm \tau^\star||},$$
so the supremum in~\eqref{eqn:first_dombry_eyi} and~\eqref{eqn:second_dombry_eyi} can be increased as follows:
\begin{align}
\alpha _{k,l}(n)& \leq 2kl\sup \left\{2-\Theta_Z (\bm{h},u):\max\{||\bm h - u\bm\tau^\star||, u\} \geq \frac{n}{1 + ||\bm \tau^\star||}\right\},\label{eqn:weaker_condition_sup}\\
\alpha _{k,\infty }(n)& \leq 2k^2\sum_{m=n}^{\infty }N(m)\sup
\left\{2-\Theta_Z (\bm{h},u):\max\{||\bm h - u\bm\tau^\star||, u\} \geq \frac{m}{1 + ||\bm \tau^\star||}\right\}.\nonumber
\end{align}%

By Item~(b) in Lemma~\ref{lem:X_mixing}, one has for all $\bm h\in \R^2$ and $u \geq 0$,
\begin{equation}
\Theta_Z (\bm{h},u)\geq \Theta_Z (\bm{h}-u\bm{\tau^\star },0)=2\Phi %
\bigg(\sqrt{\frac{\gamma(\bm{h}-u\bm{%
\tau^\star })}{2}}\bigg)>2\bigg(1-\frac{e^{-\gamma (\bm{%
h}-u\bm{\tau^\star })/4}}{\sqrt{\pi \gamma (\bm{%
h}-u\bm{\tau^\star })}}\bigg), \label{eqn:bound_phi_tail}
\end{equation}%
where the last inequality holds since for any $x\in \R$, 
\begin{equation*}
1-\Phi (x)<\frac{e^{-x^{2}/2}}{x\sqrt{2\pi }}.
\end{equation*}
Likewise, Item~(a) in Lemma~\ref{lem:X_mixing} provides $\Theta_Z(\bm h, u) \geq 2 - (a^\star)^u$. This, combined with~\eqref{eqn:bound_phi_tail}, yields
$$2 - \Theta_Z(\bm h, u) \leq \min\left \{ (a^\star)^u, 2\frac{e^{-\gamma(\bm{%
h}-u\bm{\tau^\star })/4}}{\sqrt{\pi \gamma (\bm{%
h}-u\bm{\tau^\star })}} \right \}.$$
To summarize, for $\bm{h}$ and $u$ satisfying the condition in~\eqref{eqn:weaker_condition_sup}, i.e., $\max\{||\bm h - u\bm\tau^\star||, u\} \geq n /(1 + ||\bm \tau^\star||)$, then it holds that
\begin{equation}\label{eqn:huge}
2 - \Theta_Z(\bm h, u) \leq \max\left\{(a^\star)^{n /(1 + ||\bm \tau^\star||)},
2\,\frac{\exp\left(-\frac14\left(\frac{n}{(1 + ||\bm \tau^\star||)\kappa^\star}\right)^{2H^\star}\right)}{\sqrt{\pi} \left(\frac{n}{(1 + ||\bm \tau^\star||)\kappa^\star}\right)^{H^\star}}
\right\},
\end{equation}
by replacing the semivariogram by its expression. This effectively bounds the $\alpha$-mixing coefficients.
Both expressions on the right-hand side of~\eqref{eqn:huge} tend to 0 as $n \to \infty$, the slower of which determines the rate at which the $\alpha$-mixing coefficients tend to 0. Plugging back the left hand-side of~\eqref{eqn:huge} into~\eqref{eqn:first_dombry_eyi} and~\eqref{eqn:second_dombry_eyi}, one finds that the rate is dominated by the exponential decay, and upon taking the negative logarithms of each expression, one obtains~\eqref{eqn:statement_of_lemma_3}.
\end{proof}

Next, we prove some moment conditions that will be essential in the proof of
Theorem \ref{thm:normality}.

\begin{lemm}
\label{lem:score_L3} Suppose Assumption \ref{ass:case1} holds. Then for any $%
q>0$, 
\begin{equation}
\E[||\nabla_{\bm\psi} g_{ r, p}(\bm{s}_1, 1;
{\bm\psi}^\star)||^q]<\infty  \label{eqn:third_moment_score}
\end{equation}%
and 
\begin{equation*}
\E[\sup_{{\bm\psi}\in\Psi_\varepsilon}||\nabla_{\bm\psi}^2 g_{ r,
p}(\bm{s}_1, 1; {\bm\psi})||]<\infty .
\end{equation*}
\end{lemm}

\begin{proof}
First, we show \eqref{eqn:third_moment_score}. By~\eqref{eqn:grp_def}, it suffices to show that 
\begin{equation*}
\E[||\nabla_{\bm\psi} \log f_{ \bm{h},u}(Z_1, Z_2;
{\bm\psi}^\star)||^q]<\infty ,
\end{equation*}%
where $Z_{1}$ and $Z_2$ denote $Z(\bm{s}_{1},1)$ and $Z(\bm{s%
}_{1}+\bm{h},1+u)$ for arbitrary $\bm{h}\in \mathcal{H}_{r}$ and $%
u\in \{1,\ldots,p\}$.

Firstly, using the same notation as in Lemma \ref{lem:log_f_integrable},
notice that $-V$ and $V_{1}V_{2}-V_{12}$ are both linear combinations of
terms which are asymptotically equivalent to 
\begin{equation}
\pm \frac{\log ^{K_{1}}(z_{1})\log ^{K_{2}}(z_{2})\varphi (q_{1})^{K_{3}}}{%
z_{1}^{K_{4}}z_{2}^{K_{5}}}  \label{eqn:K1to5}
\end{equation}%
as either $z_{1}$ or $z_{2}$ approach $0$ or $\infty $, for various choices
of $K_{1},K_{2},K_{3},K_{4},K_{5}\in \N$. Since 
\begin{equation*}
\nabla _{\bm{\psi }}\log f_{\bm{h},u}(z_{1},z_{2};{\bm{%
\psi }})=-\nabla _{\bm{\psi }}V+\frac{\nabla _{\bm{\psi }%
}(V_{1}V_{2}-V_{12})}{V_{1}V_{2}-V_{12}},
\end{equation*}%
we need only consider the effect of $\nabla _{\bm{\psi }}$ on the
terms in \eqref{eqn:K1to5}.

From the computations in Section~\ref{sec:density_function}, it follows
that the magnitude of the gradient with respect to ${\bm{\psi }}$ of
any term in the form of \eqref{eqn:K1to5} is asymptotically equivalent
to another term in the form of \eqref{eqn:K1to5} with $K_{3}$
unchanged, and possibly increased $K_{1}$, $K_{2}$, $K_{4}$, and $K_{5}$%
. Thus, the numerator and denominator of 
\begin{equation*}
\frac{\nabla _{\bm{\psi }}(V_{1}V_{2}-V_{12})}{V_{1}V_{2}-V_{12}}
\end{equation*}%
are both in the form of \eqref{eqn:K1to5} and thus the ratio is asymptotically equivalent to 
\begin{equation*}
\bigg|\frac{\log ^{\Delta K_{1}}(z_{1})\log ^{\Delta K_{2}}(z_{2})}{%
z_{1}^{\Delta K_{4}}z_{2}^{\Delta K_{5}}}\bigg|,
\end{equation*}%
where $\Delta K_{1},\Delta
K_{2},\Delta K_{4}$, and $\Delta K_{5}$ are non-negative. This implies that
there exists a sufficiently large $k_{2}\in \R^{+}$, a sufficiently small $%
k_{1}\in \R^{+}$ such that $k_{1}\leq k_{2}$, and a sufficiently large $%
k_{3}\in \R^{+}$ such that 
\begin{equation*}
\bigg{|}\bigg{|}-\nabla _{\bm{\psi }}V+\frac{\nabla _{\bm{%
\psi }}(V_{1}V_{2}-V_{12})}{V_{1}V_{2}-V_{12}}\bigg{|}\bigg{|}\leq \Big(\log
z_{1}+\frac{1}{z_{1}}\Big)^{k_{3}}\Big(\log z_{2}+\frac{1}{z_{2}}\Big)%
^{k_{3}}
\end{equation*}%
whenever $(z_{1},z_{2})\notin \lbrack k_{1},k_{2}]^{2}$.

Following the arguments given in Lemma \ref{lem:log_f_integrable} and using
the calculations in Section~\ref{sec:density_function}, 
\begin{equation*}
K_{k_{1},k_{2}}=\sup_{z_{1},z_{2},\bm{\psi }}\{||\nabla _{%
\bm{\psi }}\log f_{\bm{h},u}(z_{1},z_{2};{\bm{\psi }}%
^{\star })||:z_{1},z_{2}\in \lbrack k_{1},k_{2}],\bm{\psi }\in \Psi_\varepsilon \}
\end{equation*}%
can be shown to be finite for any $k_{1},k_{2}\in \R^{+}$ such that $%
k_{1}\leq k_{2}$. Then, using H\"{o}lder's inequality as in Lemma \ref%
{lem:log_f_integrable}, we show that 
\begin{equation*}
\E[||\nabla_{\bm\psi} \log f_{ \bm{h},u}(X_1, X_2;
{\bm\psi}^\star)||^q]<\infty ,
\end{equation*}%
which proves \eqref{eqn:third_moment_score}.

Similar arguments yield that 
\begin{equation*}
\nabla _{\bm{\psi }}^{2}\log f_{\bm{h},u}(z_{1},z_{2};{%
\bm{\psi }})=-\nabla _{\bm{\psi }}^{2}V+\frac{%
(V_{1}V_{2}-V_{12})\nabla _{\bm{\psi }}^{2}(V_{1}V_{2}-V_{12})-\big(%
\nabla _{\bm{\psi }}(V_{1}V_{2}-V_{12})\big)^{2}}{%
(V_{1}V_{2}-V_{12})^{2}}
\end{equation*}%
can be bounded in absolute value by an integrable function that is
independent of ${\bm{\psi }}$, proving the uniform integrability of $%
||\nabla _{\bm{\psi }}^{2}g_{r,p}(\bm{s}_{1},1;{\bm{\psi 
}})||$ on $\Psi_\varepsilon $.
\end{proof}

\begin{proof}[Proof of Theorem \protect\ref{thm:normality}]
The proof in \cite{davis2013b} for the asymptotic normality of their
estimator implies the asymptotic normality of ours. However, we provide a
summary of their proof for completeness.

Consider the Taylor expansion of $\nabla _{\bm{\psi }}\mathrm{PL}%
^{(m,T)}(\hat{\bm{\psi }})$ around the true parameter vector ${%
\bm{\psi }}^{\star }$:
\begin{equation*}
\nabla _{\bm{\psi }}\mathrm{PL}^{(m,T)}(\hat{\bm{\psi }}%
)=\nabla _{\bm{\psi }}\mathrm{PL}^{(m,T)}({\bm{\psi }}^{\star
})+\nabla _{\bm{\psi }}^{2}\mathrm{PL}^{(m,T)}(\Tilde{\bm{%
\psi }})(\hat{\bm{\psi }}-{\bm{\psi }}^{\star }),
\end{equation*}%
for some $\Tilde{{\bm{\psi }}}\in \Psi_\varepsilon $ whose components are
between those of ${\bm{\psi }}^{\star }$ and $\hat{\bm{\psi }}
$. Since $\mathrm{PL}^{(m,T)}({\bm{\psi }})$ is maximized by $\hat{%
\bm{\psi }}$, we can write 
\begin{equation}
\frac{1}{(m^{2}T)^{1/2}}\nabla _{\bm{\psi }}\mathrm{PL}^{(m,T)}({%
\bm{\psi }}^{\star })=\bigg(-\frac{1}{m^{2}T}\nabla _{\bm{%
\psi }}^{2}\mathrm{PL}^{(m,T)}(\Tilde{\bm{\psi }})\bigg)\bigg(%
(m^{2}T)^{1/2}(\hat{\bm{\psi }}-{\bm{\psi }}^{\star })\bigg).
\label{eqn:taylor_simplified}
\end{equation}

We recall that ${\bm{\psi }}^{\star }$ is the unique maximizer of $%
\E[g_{ r,
p}(\bm{s}_1, 1; {\bm\psi})]$. It follows from Lemma \ref%
{lem:score_L3} and the dominated convergence theorem that 
\begin{equation*}
\E[\nabla_{\bm\psi} g_{ r, p}(\bm{s}_1, 1; {\bm\psi}^\star)]%
=\nabla _{\bm{\psi }}\E[g_{ r, p}(\bm{s}_1, 1;
{\bm\psi}^\star)]=0.
\end{equation*}%
This fact, together with Lemmas \ref{lem:alpha_mixing} and \ref{lem:score_L3}%
, gives sufficient conditions to apply the central limit theorem provided in 
\cite{bolthausen1982}, which implies 
\begin{equation*}
\frac{1}{(m^{2}T)^{1/2}}\sum_{\bm{s}\in S_{m}}\sum_{t=1}^{T}\nabla _{%
\bm{\psi }}g_{r,p}(\bm{s},t;{\bm{\psi }}^{\star })%
\xrightarrow{\mathrm{d}}\mathcal{N}\big(0,\Sigma \big),\qquad m,T\rightarrow
\infty ,
\end{equation*}%
where $\Sigma $ is given by \eqref{eqn:big_sigma}.

We can repeat the arguments in the proof of Theorem \ref{thm:consistency} to
show that ${\bm{\psi }}^{\star }$ is the unique maximizer of $%
\mathcal{R}^{(m,T)}({\bm{\psi }})$. Arguments in Lemma \ref%
{lem:score_L3} justify that we can again use the central limit theorem in 
\cite{bolthausen1982} to achieve 
\begin{equation*}
\frac{1}{(m^{2}+mT)^{1/2}}\nabla _{\bm{\psi }}\mathcal{R}^{(m,T)}({%
\bm{\psi }}^{\star })\xrightarrow{\mathrm{d}}\mathcal{N}\big(0,\Tilde{%
\Sigma}\big),\qquad m,T\rightarrow \infty ,
\end{equation*}%
where $\Tilde{\Sigma}$ is a valid covariance matrix. Therefore, 
\begin{equation*}
\frac{1}{(m^{2}T)^{1/2}}\nabla _{\bm{\psi }}\mathcal{R}^{(m,T)}({%
\bm{\psi }}^{\star })\xrightarrow{\mathrm{p}}0,\qquad m,T\rightarrow
\infty .
\end{equation*}%
These results can be combined with Slutsky's lemma to yield 
\begin{equation}
\frac{1}{(m^{2}T)^{1/2}}\nabla _{\bm{\psi }}\mathrm{PL}^{(m,T)}({%
\bm{\psi }}^{\star })\xrightarrow{\mathrm{d}}\mathcal{N}\big(0,\Sigma %
\big),\qquad m,T\rightarrow \infty .  \label{eqn:score_PL_normal}
\end{equation}

Additionally, Proposition \ref{lem:X_mixing} and Lemma \ref{lem:score_L3}
provide sufficient conditions for the strong law of large numbers from \cite%
{straumann2006} to apply to $\{\nabla _{\bm{\psi }}^{2}g_{r,p}(%
\bm{s},t;{\bm{\psi }})\}_{\bm{s}\in \Z^{2},t\in \N}$.
Therefore, uniformly on $\Psi_\varepsilon $, 
\begin{equation*}
-\frac{1}{m^{2}T}\sum_{\bm{s}\in S_{m}}\sum_{t=1}^{T}\nabla _{%
\bm{\psi }}^{2}g_{r,p}(\bm{s},t;{\bm{\psi }})\overset{%
\mathrm{a.s.}}{\longrightarrow }\E[-\nabla_{\bm\psi}^2 g_{ r,
p}(\bm{s}_1, 1; {\bm\psi})],\qquad m,T\rightarrow \infty ;
\end{equation*}%
similarly, 
\begin{equation*}
-\frac{1}{m^{2}T}\nabla _{\bm{\psi }}^{2}\mathcal{R}^{(m,T)}({%
\bm{\psi }})\overset{\mathrm{a.s.}}{\longrightarrow }0,\qquad
m,T\rightarrow \infty ,
\end{equation*}%
and thus, 
\begin{equation*}
-\frac{1}{m^{2}T}\nabla _{\bm{\psi }}^{2}\mathrm{PL}^{(m,T)}({%
\bm{\psi }})\overset{\mathrm{a.s.}}{\longrightarrow }%
\E[-\nabla_{\bm\psi}^2 g_{ r, p}(\bm{s}_1, 1; {\bm\psi})]%
,\qquad m,T\rightarrow \infty .
\end{equation*}%
Since the convergence is uniform on $\Psi_\varepsilon $, and that $\Tilde{\bm{%
\psi }}\rightarrow {\bm{\psi }}^{\star }$ almost surely from the
consistency of $\hat{\bm{\psi }}$, we have 
\begin{equation*}
-\frac{1}{m^{2}T}\nabla _{\bm{\psi }}^{2}\mathrm{PL}^{(m,T)}(\Tilde{%
\bm{\psi }})\overset{\mathrm{a.s.}}{\longrightarrow }%
F = \E[-\nabla_{\bm\psi}^2 g_{ r, p}(\bm{s}_1, 1;
{\bm\psi}^\star)],\qquad m,T\rightarrow \infty .
\end{equation*}%
By combining this result with~\eqref{eqn:taylor_simplified} and~\eqref{eqn:score_PL_normal} and using Slutsky's lemma, we finally prove the theorem.
\end{proof}

\section{Simulation study and justification of assumptions}
\subsection{Simulation strategy}\label{Sec_Appendix_Simulations}

Here we outline a method for simulating realizations of our model, defined recursively in~\ref{eqn:max_auto}. The recurrence relation serves to reduce computational complexity, in that it suffices to simulate independent replications of the spatial random field $W$ in~\eqref{Eq_InnovationField} on the two-dimensional grid $\mathcal S_m$ in~\eqref{eqn:Sm} for $m\in \N^+$.

To leverage \eqref{eqn:max_auto} when simulating our field at the space-time coordinate $(\bm s, t) \in \mathcal S_m\times \N^+$, one needs the value of the field at $(\bm s - \bm\tau, t-1)$. This limits the permissible values of the advection parameter $\bm\tau$ when simulating our space-time field on all of $\mathcal S_m \times \{1,\ldots, T\}$, since $\bm\tau$ must be aligned with the grid of simulation sites. If this is the case, the simulation method is a trivial application of~\eqref{eqn:max_auto}. The main practical consideration to keep in mind is that information from outside the domain $\mathcal S_m$ ``drifts'' inwards at a speed of $\bm\tau$, and so the innovation field $W_{\tilde t}$ should be simulated sufficiently far outside of $\mathcal S_m$ for each $\tilde t \in \N^+$ less than $t$.

\begin{Rq}
This simulation scheme should be used cautiously when performing inference on the random field using the pairwise likelihood estimation strategy described in Section~\ref{Sec_InferenceSimulations}. We require that $\bm h/u\neq \bm\tau^\star$ for any pair $u\in\{1,\ldots,p\}$ and $\bm h \in \mathcal{H}_r$. However, it holds by construction that $\bm \tau^\star$ is aligned with the grid of simulation sites. An appropriate solution is to simulate on a finer spatial grid than the grid of spatial lags in the design mask $\mathcal H_r$ used for the estimation.
\end{Rq}

\subsection{The ratio random field as a diagnostic tool}
\label{App_Diagnos_ratio}

We now present a method to verify that $\bm\tau \notin \{\bm{h}/u : \bm{h}\in\mathcal{H}_r, u=1,\ldots,p\}$. For $\bm{h}\in\mathcal{H}_r$ and $u\in\{1,\ldots,p\}$, consider the ratio random field
\begin{equation}\label{eqn:chi_def}
\chi_{\bm{h}, u}(\bm{s},t) = \bigg[\frac{Z(\bm{s}+\bm{h},t+u)}{Z(\bm{s},t)}\bigg]^{1/u},\qquad \bm{s}\in\R^2,t\in\N,
\end{equation}
where our space-time model $Z$ in~\eqref{eqn:max_auto}
is assumed to be space-time mixing (see Definition~\ref{def:mixing_in_space_and_time}).

\begin{figure}
    \centering
    \begin{subfigure}[b]{0.495\textwidth}
        \centering
        \includegraphics[width=\textwidth]{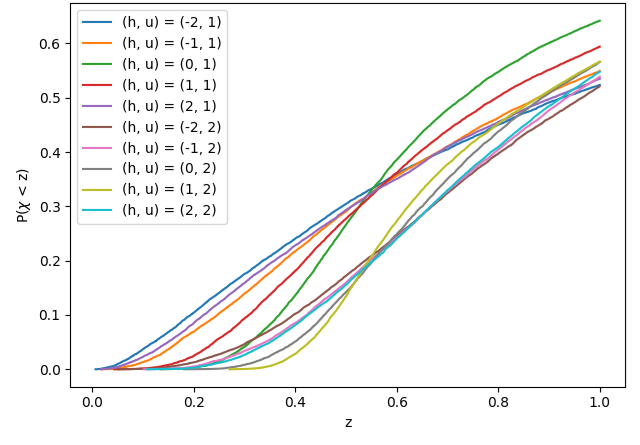}
        \caption
        {{\small $\kappa = 1.2$, $H=0.7$, $\tau = 1/3$, $a=0.5$}}    
    \end{subfigure}
    \hfill
    \begin{subfigure}[b]{0.495\textwidth}  
        \centering 
        \includegraphics[width=\textwidth]{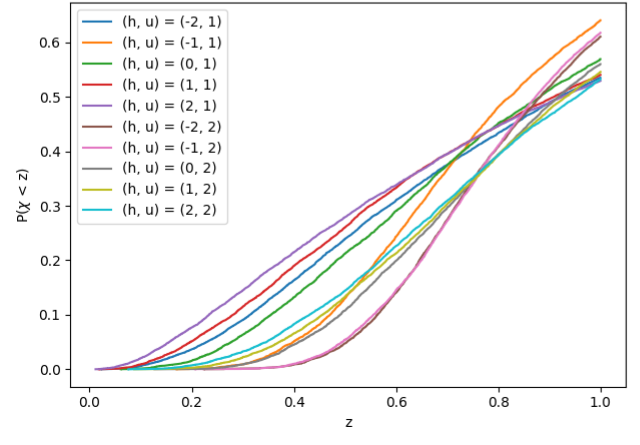}
        \caption
        {{\small $\kappa = 1.5$, $H=0.5$, $\tau = -0.75$, $a=0.7$}}
    \end{subfigure}
    \vskip\baselineskip
    \begin{subfigure}[b]{0.495\textwidth}   
        \centering 
        \includegraphics[width=\textwidth]{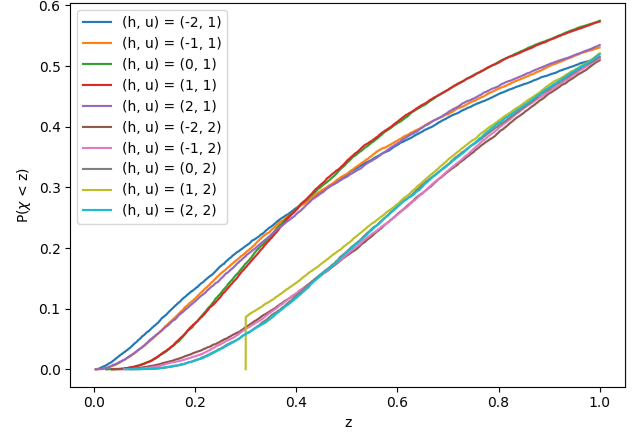}
        \caption
        {{\small $\kappa = 1.0$, $H=0.6$, $\tau = 0.5$, $a=0.3$}}
    \end{subfigure}
    \hfill
    \begin{subfigure}[b]{0.495\textwidth}   
        \centering 
        \includegraphics[width=\textwidth]{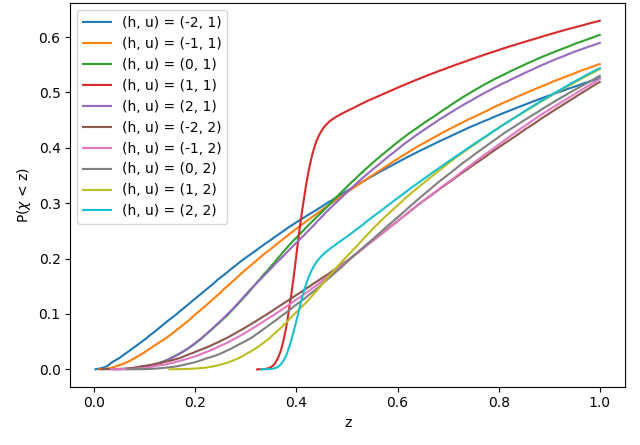}
        \caption
        {{\small $\kappa = 1.7$, $H=0.8$, $\tau = 0.95$, $a=0.4$}}
    \end{subfigure}
    \caption{Empirical distribution of the margins $\chi_{h, u}(s,t)$ computed from a single realization are shown on $[0,1]$ for $(h,u)$ such that $h\in\{-2,-1,0,1,2\}$ and $u = 1, 2$. The four plots correspond to four different parametrizations of our space-time model as indicated below each plot.}
    \label{fig:chi_plots}
\end{figure}

\begin{prop}\label{prp:empirical_taus}
Suppose that the spatial random field $W$ is Brown--Resnick with exponent measure given in~\eqref{eqn:bivariate_exponent_measure_BR}. It holds that
\begin{equation}\label{eqn:inf_a}
    \inf \left\{\chi_{\bm{h}, u}(\bm{s},t) : \bm s\in \R^2, t \in \N \right\} = a
\end{equation}
almost surely, if and only if $\bm\tau=\bm{h}/u$. Moreover, in this case, for any $\bm{s}\in\R^2$ and $t\in\N$, $\P\big(\chi_{\bm{h}, u}(\bm s,t)=a\big) = a^u$. Otherwise, when $\bm\tau\neq\bm{h}/u$, $$\inf \left\{\chi_{\bm{h}, u}(\bm{s},t): \bm s\in \R^2, t \in \N \right\}=0$$
almost surely.
\end{prop}
\begin{proof}
\begin{description}
\item[Case 1:] $\bm\tau=\bm{h}/u$. For all $\bm{s}\in\R^2$ and $t\in\N$,
\begin{equation*}
    \chi_{\bm{h}, u}(\bm{s},t) \stackrel{\mathrm{d}}{=} \max\bigg\{a, \Big(\big(1-a^u\big)R\Big)^{1/u}\bigg\}
\end{equation*}
follows directly from \eqref{eqn:max_auto}, where $R$ is the ratio of two independent exponential random variables each with unit rate. The field $Z$ is space-time mixing, and since  $\chi_{\bm{h}, u}$ is composed of local operations of $Z$, it is also space-time mixing, proving \eqref{eqn:inf_a}.
Moreover, for any $\bm{s}\in\R^2$ and $t\in\N$,
\begin{equation*}
    \P\big(\chi_{\bm{h}, u}(\bm s,t)=a\big) = \P\bigg(R \leq \frac{a^u}{1-a^u}\bigg) = a^u.
\end{equation*}
\item[Case 2:] $\bm\tau \neq \bm{h}/u$.
It follows directly from the bivariate exponent measures in~\eqref{eqn:bivariate_exponent_measure_BR} that for any $\bm{s}\in\R^2$ and $t\in\N$, the random variable $\chi_{\bm{h}, u}(\bm s,t)$ assigns a non-zero probability measure to the interval $(0,\varepsilon)$ for any $\varepsilon > 0$. Since $\chi_{\bm{h}, u}$ is space-time mixing, the result follows.
\end{description}
\end{proof}

Proposition~\ref{prp:empirical_taus} highlights that the distribution function of the margins of $\chi_{\bm h, u}$ carries information about the decay parameter $a$ whenever $\bm h / u = \bm\tau$. In practice, for some $\bm h \in \mathcal H_r$ and $u \in \{1\ldots,p\}$, the empirical distribution function of the margins of $\chi_{\bm h, u}$, given by
\begin{equation}\label{eqn:empirical_chi_curves}
    \hat{F}_{\chi_{\bm{h},u}}(z) = \frac{1}{\vert\mathcal{S}_{m} \cap \mathcal{S}_m - \bm{h}\vert\times (T-u)}\sum_{\bm{s}\in \mathcal{S}_{m} \cap \mathcal{S}_m - \bm{h}}\sum_{t=1}^{T-u}\mathbb{I}(\chi_{\bm{h},u}(\bm{s},t) \leq z),
\end{equation}
can be computed, and cases where $\bm h / u = \bm\tau$ can be identified. Indeed, when $\bm h / u = \bm\tau$, the empirical distribution function in~\eqref{eqn:empirical_chi_curves} is 0 on the interval $(0,a)$, then it jumps to the value $a^u$ at $a$. This behavior of the empirical distribution function indicates $a$ as the jump location, and $\bm\tau$ as $\bm h/u$. If $a$ and $\bm\tau$ are identified in this way, then the pairwise likelihood function in~\eqref{eqn:pl-likelihoood_spatial} can be used to estimate the spatial parameter $\bm \theta$ to finish the estimation procedure.

To illustrate this in a numerical example, we simulate $Z$ in~\eqref{eqn:max_auto} with a Brown--Resnick spatial dependence structure on a one dimensional spatial domain for computational efficiency. Indeed, $Z(s, t)$ is simulated for $s, t \in \{1,\ldots,100\}$ using the simulation strategy described above for four different parameter vectors $\bm \psi = (\kappa, H, \tau, a)^{\prime}$, where $\kappa$ and $H$ parametrize the semivariogram $\gamma(x) = \left( | x |/\kappa \right)^{2H}$, with $x \in \mathbb{R}$. In a following step, from one realization of our field with parameter vector $\bm \psi$, we compute the corresponding ratio random fields $\chi_{h, u}$ for $h \in \{-2,-1,0,1,2\}$ and $u = 1, 2$. In Figure~\ref{fig:chi_plots}, we see that the empirical distribution function in~\eqref{eqn:empirical_chi_curves} of the margins of the random fields are visually informative for the temporal parameters $a$ and $\tau$. Indeed, when $\tau \approx h/u$, there is a jump in the empirical distribution function at the value $a$ as stated previously.


\subsection{Diagnostics}
\label{App_Diagnos}


In the following, we use the ratio random field to perform a preliminary examination of the hourly wind gust data. This step justifies the assumption that $\bm\tau\in \Psi_\varepsilon$, where $\Psi_\varepsilon$ in~\eqref{eqn:psi_epsilon} is a parameter space that excludes vectors $\bm\psi$ with $\|\bm\tau - \bm h/u\| < \varepsilon$ for all $\bm h \in \mathcal H_r$ and $u = 1,\ldots,p$. The empirical distribution functions plotted in Figure~\ref{Fig_DataCurves} for several of these space-time lags do not exhibit any clear discontinuities on the interval $(0,1)$, in contrast with the discontinuous curve in Figure~\ref{fig:chi_plots}~(c) with $\tau=h/u$. Under the assumption that the data follow our model, the smoothness of the curves in Figure~\ref{Fig_DataCurves} indicates that $\bm\tau^\star\neq \bm h / u$ for all considered $\bm h$ and $u$. 

\begin{figure}
    \centering
    \includegraphics[scale=0.5]{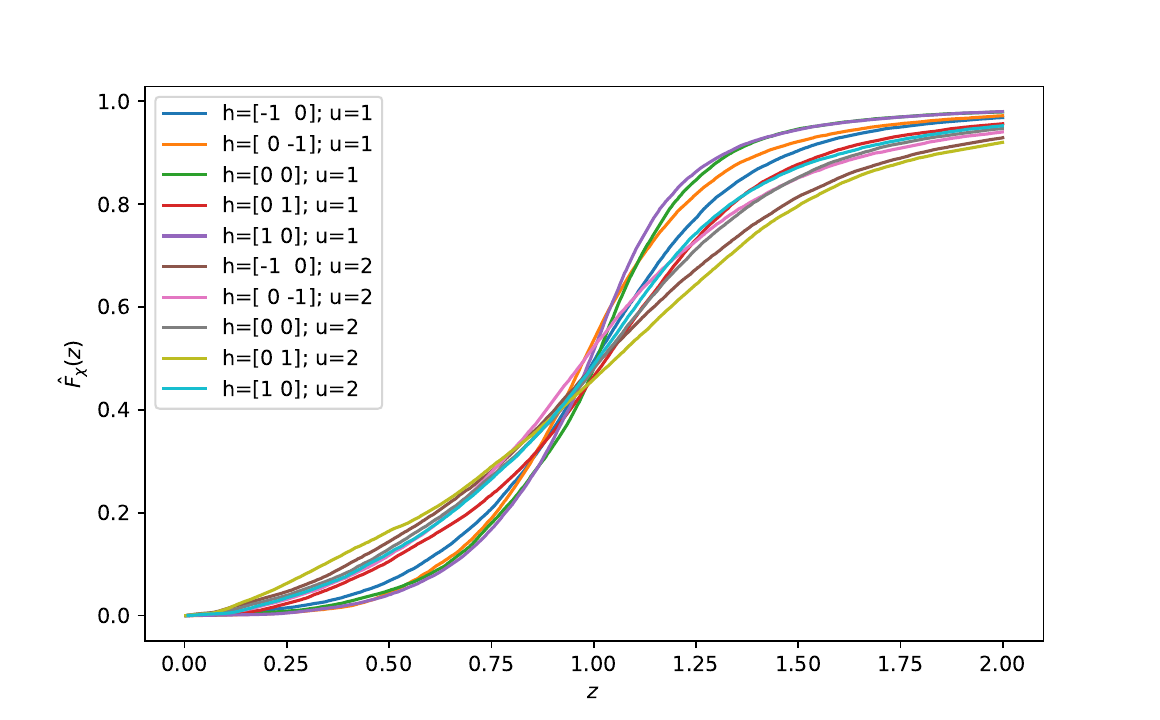}
    \caption{Empirical distribution functions of the margins of the ratio random fields $\chi_{\bm h, u}$, for several space-time lags $(\bm h, u)$, in~\eqref{eqn:empirical_chi_curves}, evaluated from the hourly wind gust data.}
    \label{Fig_DataCurves}
\end{figure}

\section{Single-site marginal parameters}\label{sec:single_site_params}

\begin{figure}[H]
    \centering
    \includegraphics[scale = 0.5]{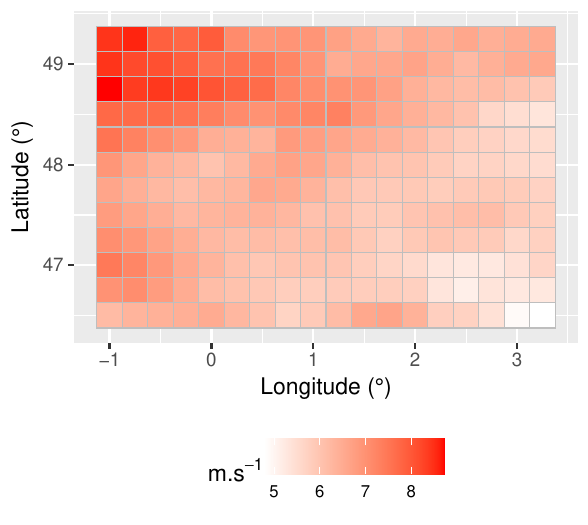}
    \includegraphics[scale = 0.5]{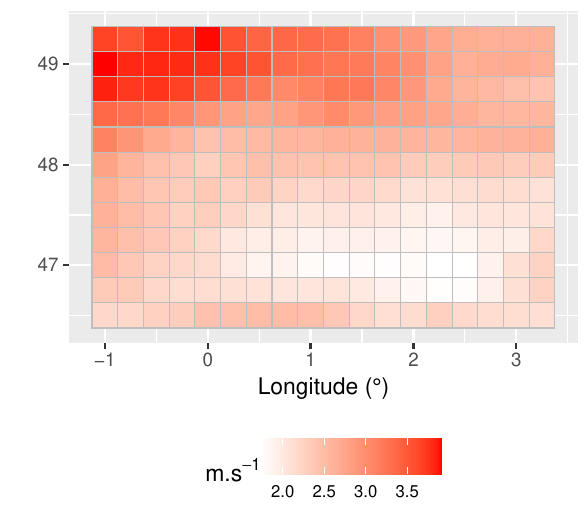}
    \includegraphics[scale = 0.5]{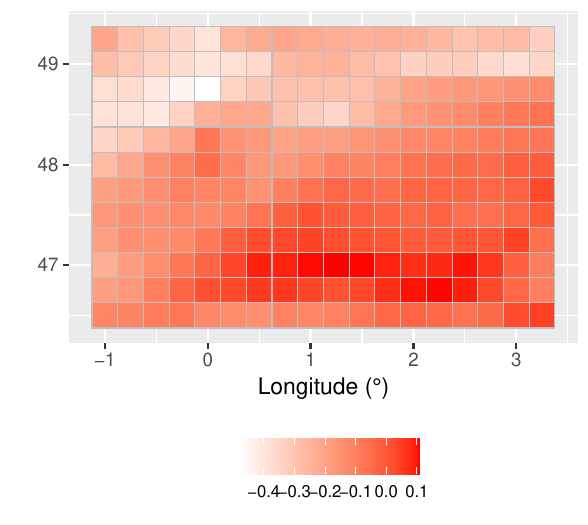}
    \caption{Estimated GEV parameters (location, scale and shape from left to right) used in the marginal rescaling in Section~\ref{Sec_CaseStudy}.}
    \label{fig:gev_params}
\end{figure}

\end{document}